\documentclass[journal]{IEEEtran}
\usepackage{amsmath,amsfonts}
\usepackage{amssymb}
\usepackage{amsthm}
\usepackage{rotating}
\usepackage{multirow}
\usepackage{color}
\usepackage{algorithmic}
\usepackage[ruled]{algorithm}
\usepackage{array}
\usepackage{subfigmat}
\usepackage{cite}
\usepackage{url}
\usepackage{graphicx}
\usepackage{bm}
\usepackage{booktabs}
\usepackage{textcomp}
\usepackage{verbatim}

\usepackage[update, suffix=]{epstopdf}

\makeatletter
\newcommand\fs@algorithm{%
	\def\@fs@pre{\hrule height 1pt depth0pt\kern2pt}%
	\def\@fs@mid{\hrule height 1pt \kern2pt}%
	\def\@fs@post{\kern2pt\hrule height 1pt\relax}%
	\let\@fs@iftopcapt\iftrue}
	\def\@fs@cfont{\normalfont\small\bfseries} 
	\let\@fs@capt\floatc@plain  
\makeatother
\floatstyle{algorithm}
\restylefloat{algorithm}

\DeclareFontFamily{U}{mathx}{}
\DeclareFontShape{U}{mathx}{m}{n}{<-> mathx10}{}
\DeclareSymbolFont{mathx}{U}{mathx}{m}{n}
\DeclareMathAccent{\widecheck}{0}{mathx}{"71}

\def\x{{\mathbf x}}

\def\y{{\mathbf y}}

\def\z{{\mathbf z}}

\def\n{{\mathbf n}}

\def\B{{\mathbf B}}
\def\C{{\mathbf C}}
\def\D{{\mathbf D}}

\def\L{{\mathbf L}}
\def\R{{\mathbb R}}

\def\I{{\mathbf I}}
\def\P{{\mathbf P}}

\def\E{{\mathbf E}}

\newcommand{\argmin}{\mathop{\mathrm{argmin}}\limits}

\makeatletter
\def\th@myStyle{
	\thm@headfont{\itshape} 
	\thm@notefont{\itshape} 
	\thm@headpunct{:}                  
	\renewcommand\@upn{\textit}
	\normalfont 
}
\makeatother

\theoremstyle{myStyle} 
\newtheorem{Ex}{Example}
\newtheorem{Rem}{Remark}

\newtheorem{theorem}{Theorem}
\newtheorem{prop}{Proposition}

\newtheorem{lemma}{Lemma}
\def\thline{\noalign{\hrule height 1pt}}

\begin{document}
\bstctlcite{IEEEexample:BSTcontrol} 

\title{Sparsity-Enhanced Multilayered Non-Convex Regularization with Epigraphical Relaxation \\for Debiased Signal Recovery}
%
\author{Akari~Katsuma, 
	Seisuke~Kyochi,~\IEEEmembership{Member,~IEEE,}
	Shunsuke~Ono,~\IEEEmembership{Senior Member,~IEEE,}
	and~Ivan~Selesnick,~\IEEEmembership{Fellow,~IEEE}
	\thanks{S. Kyochi is with Kogakuin University.}
	\thanks{A. Katsuma and S. Ono are with Tokyo Institute of Technology.}
	\thanks{I. Selesnick is with New York University.}
	\thanks{This work was supported by JSPS Grants-in-Aid (21K04045, 22H03610, 22H00512, 23H01415, 23K17461, 23K24866, 24K07481) and ROIS NII Open Collaborative Research 2024 (24S0108).}
	\thanks{Manuscript received XXX XX, XXXX.}
}

\markboth{Journal of \LaTeX\ Class Files,~Vol.~14, No.~8, August~2021}%
{Shell \MakeLowercase{\textit{et al.}}: A Sample Article Using IEEEtran.cls for IEEE Journals}


\maketitle

\begin{abstract}
This paper proposes a precise signal recovery method with multilayered non-convex regularization, enhancing sparsity/low-rankness for high-dimensional signals including images and videos. In optimization-based signal recovery, multilayered convex regularization functions based on the $\ell_1$ and nuclear-norms not only guarantee a global optimal solution but also offer more accurate estimation than single-layered ones, thanks to their faithful modeling of structured sparsity and low-rankness in high-dimensional signals. However, these functions are known to yield biased solutions (estimated with smaller amplitude values than the true ones). To address this issue, multilayered non-convex regularization functions have been considered, although they face their own challenges: 1) their closed-form proximity operators are unavailable, and 2) convergence may result in a local optimal solution. In this paper, we resolve the two issues with an approach based on epigraphical relaxation (ER). First, ER decomposes a multilayered non-convex regularization function into the outermost function and epigraph constraints for the inner functions, facilitating the computation of proximity operators. Second, the relaxed regularization functions by ER are integrated into a non-convexly regularized convex optimization model to estimate a global optimal solution with less bias. Numerical experiments demonstrate the bias reduction achieved by the proposed method in image recovery and principal component analysis.
\end{abstract}

\begin{IEEEkeywords}
	Convex optimization, signal recovery, epigraphical projection, LiGME model
\end{IEEEkeywords}

\section{Introduction}
\label{sec:intro}
\IEEEPARstart{S}{ignal} recovery, such as denoising, deburring, interpolation, and compressed sensing reconstruction, is a crucial task in signal processing and machine learning. One of the fundamental methods for signal recovery is the optimization-based approach \cite{Combettes2005,Combettes2008,Wen2018} that solves an optimization problem carefully-modeled with the prior information of the target signal, such as sparsity and/or low-rankness. An exact evaluation of these properties can be achieved by introducing non-convex regularization functions ({\it{regularizers}}), such as the $\ell_0$ pseudo-norm and rank. However, non-convex optimization employing these regularizers leads to a local optimal solution depending on initial values. Moreover, the $\ell_0$ pseudo-norm/rank minimization problem is known to be NP-hard\cite{Natarajan1995}. Therefore, the convexly-relaxed regularizers, such as the $\ell_1$-norm and the nuclear-norm, have widely been employed in sparse/low-rank signal recovery\cite{Starck2010,Candes2011}.

Convexly-relaxed regularizers ensure convergence to a global optimizer; however, compared with non-convex regularizers, they tend to more strongly penalize high-amplitude components of input signals, and thus, result in biased solutions with underestimation on their high-amplitude components. To circumvent this problem, convexity-preserving non-convex regularizers (CPNRs) have attracted attention\cite{SelesnickGMC2017,Yin2019,Cai2020,Abe2020,Yata2022,Chen2021,Yukawa2023,Chen2023,Al-Shabili2021,Zhang2023}. Unlike standard non-convex regularizers\cite{Donoho1994,Blumensath2008,Chartrand2007,Fan2001,Woodworth2016,Beren1997,Mesbahi1997,Nie2012,Hu2013}, CPNRs not only reduce the bias but also guarantee convergence to a global optimal solution (see Sec. \ref{sec:relatedWork2}). 

As a different perspective from the convexity or non-convexity of regularizers, multilayered composite functions have attracted attention for accurate recovery of high dimensional signals (e.g., images and videos), to characterize complex properties often inherent in high-dimensional signals, e.g., group-sparsity/low-rankness in spatial, spectral, and time domains. In the past several decades, group sparsity/low-rankness-aware multilayered regularizers have been developed based on convexly-relaxed regularizers mostly as two-layered ones, e.g., total variation (TV) \cite{Rudin1992, Bresson2008, Chan2010, Ono2014DVTV} with the sum of the $\ell_2$-norms (the $\ell_{2,1}$-norm) and structure-tensor TV (STV)\cite{Bredies2010,Lefkimmiatis2015,Chierchia2014,Ono2016} with the sum of the nuclear-norms, and achieved better accuracy in high-dimensional signal recovery than single-layered ones thanks to their sophisticated prior modeling.

These developments motivate us to realize multilayered (three or even more layered) CPNRs for more accurate high-dimensional signal recovery with less bias than the multilayered convex regularizers. However, the conventional CPNRs have limited capability in handling multilayered regularizers. In fact, the conventional studies have mainly focused on single or two-layered regularizers (the $\ell_1$-norm\cite{SelesnickGMC2017,Al-Shabili2021,Abe2020,Yata2022,Zhang2023}, the nuclear-norm\cite{Yin2019,Abe2020,Al-Shabili2021}, and the grouped $\ell_{2,1}$-norm\cite{Zhang2023}) with the closed-form proximity operators, for one and two-dimensional signal recovery. As the number of layers in a multilayered composite function increases, it is more difficult to find the closed-form proximity operator of the function ({\it{non-proximable}}), which limits the degree of freedom for designing multilayered regularizers. Consequently, the conventional CPNRs with limited layers fail to fully realize their potential in effectively capturing the structured properties of high-dimensional signals.

This paper proposes multilayered non-convex regularization by epigraphical-relaxation (ER) \cite{Kyochi2021}, and develops a new CPNR mitigating the limitation on the number of layers to enhance the accuracy of high-dimensional signal recovery. ER decomposes a multilayered regularizer into the outermost function and the constraints given from the inner functions, and makes the multilayered regularization problem tractable as long as each proximity operator is available, even if the multilayered regularizer is non-proximable. As practical applications of our method, we apply it to image recovery tasks and principal component analysis.

\subsection{Contribution}\label{sec:contribution}
The contributions of this paper are fourfold.
\begin{enumerate}
	\item{
	We extend the conventional ER technique for convex optimization problems to non-convex ones with
 guaranteeing the identical global optimal solution set, under the {\it{strictly increasing}} property as detailed in Sec. \ref{sec:ERNCR}.}
	\item{We establish a new model integrating the linearly-involved generalized Moreau-enhanced (LiGME) model and ER (termed ER-LiGME model) that enables us to use multilayered non-convex regularization in a convex optimization. As it will be discussed, the straightforward integration of the LiGME model and the ER technique cannot achieve non-convexification of the seed function. This work shows that {\it{guided observation extension}} (GOE) in the ER-LiGME model can integrate the LiGME method and the ER technique as shown in Sec. \ref{subsec:ERLiGME}. 
	}
	\item{Our ER-LiGME model also contributes to significantly simplify the overall convexity condition, compared with the conventional LiGME models\cite{Abe2020, Yata2022, Yukawa2023, Chen2021,Chen2023}. Thanks to this, the design of our LiGME penalty can avoid the computationally expensive procedures in high-dimensional signal recovery, such as large scale inverse matrix calculation, LDU decomposition, and singular or eigen value decomposition suggested in the conventional LiGME models\cite{Abe2020, Chen2021, Yata2022, Yukawa2023,Chen2023}.
	}
	\item{We provide two practical examples of multilayered non-convex regularization: the ER-LiGME DSTV and ASNN regularizers, which are the non-convexification of the DSTV and ASNN regularizers proposed in \cite{Kyochi2021}, for compressed image sensing and robust principal component analysis (RPCA)\cite{Candes2011} applications respectively. In both applications, experimental results show that the proposed regularizers not only achieves better accuracy of sparse signal recovery but also avoids biased signal estimation. 
	}
\end{enumerate}
In the preliminary work of this paper \cite{Katsuma2023}, our discussion was limited to the integration of ER and LiGME model, whereas this paper extends the discussion to include broader multilayered non-convex regularization problems and the preservation of global optimal solutions using and not using the ER. As a further advancement, this paper newly provides an application of the ER-LiGME ASNN regularizer in RPCA.

\subsection{Related Work}\label{sec:relatedWork}
This section reviews conventional non-convex regularizers for signal processing, not limited to signal recovery.

\subsubsection*{Standard Non-convex Regularizers}\label{sec:relatedWork1}
Non-convex regularizers for promoting sparsity include $\ell_0$ pseudo-norm\cite{Donoho1994,Blumensath2008}, $\ell_p$-norm ($0<p<1$)\cite{Chartrand2007}, smoothly clipped absolute deviation (SCAD)\cite{Fan2001}, and $\widetilde{p}$-shrinkage\cite{Woodworth2016}. 
On the other hand, low-rankness-promoting non-convex regularizers include rank\cite{Beren1997,Mesbahi1997}, Shatten-$p$-norm\cite{Nie2012}, and truncated nuclear-norm\cite{Hu2013}.
In addition, multilayered non-convex regularizers that consider structured-sparsity have been proposed, e.g., $\ell_{2,0}$-norm\cite{Cai2013}, $\ell_{2,p}$-norm\cite{Wang2018} and capped $\ell_{2,1}$-norm\cite{Ma2018}. These regularizers lose the convexity of the optimization problem, and thus, a global optimal solution is not guaranteed.

\subsubsection*{Convexity-preserving Non-convex Regularizers}\label{sec:relatedWork2}
CPNRs are the generalized minimax concave (GMC)\cite{SelesnickGMC2017,Yin2019,Cai2020} and the LiGME penalty functions\cite{Abe2020,Yata2022,Chen2021,Yukawa2023,Chen2023}, the sharpening sparse regularizer (SSR)\cite{Al-Shabili2021}, and the partially smoothed difference-of-convex (pSDC) regularizer\cite{Zhang2023}. These penalties are converted from a convex regularization function, called a \textit{seed function}\cite{Yukawa2023}, into a sparsity-enhanced non-convex regularizer by subtracting its (generalized) Moreau envelope with carefully chosen parameters to ensure the overall convexity condition of the cost function.
The LiGME penalty generalizes the GMC one with the $\ell_1$ seed function to the one with an arbitrary convex seed function, involving linear operators. The SSR generalizes the $\ell_2$-norm  employed in the GMC and LiGME penalties as a smoothing function so that we can choose other smoothing functions to adjust the shape of non-convex regularizers. Additionally, the pSDC regularizer has been proposed for a broader class of CPNR models, which can handle general convex data fidelity terms and constraints, and multiple regularizers. It opens up new applications that have not been possible with the conventional CPNRs, such as Poisson noise reduction\cite{Zhang2023}. The minimization problems with these penalties are ensured to converge to a global optimal solution. However, since the existing CPNRs rely on computing the proximity operator of the seed function, they cannot handle non-proximable multilayered regularization.

\section{Preliminaries}
\label{sec:Prelim}
\begingroup 
\begin{table}[t]
	\caption{Basic notations.}
	\label{tab:notations}
	\begin{tabular}{>{\centering\arraybackslash}p{13.5em}|>{\centering\arraybackslash}p{15em}}
		\thline
		Notation & Terminology \\
		\thline
		
		$\mathbb{N}$, $\mathbb{C}$, $\mathbb{R}$ and $\mathbb{R}_{++}$ &
		\begin{tabular}{c}
			Natural, complex, real, and \\nonnegative real
			numbers
		\end{tabular}\\ \hline
		
		
		$\mathcal{A}^{N},\ \mathcal{A}^{N_r\times N_c}$ &
		\begin{tabular}{c}
			\scriptsize
			$N$-dimensional, $N_r\times N_c$-dimensional\\ vectors/matrices \\ with elements in $\mathcal{A} \subset\mathbb{R}$
		\end{tabular}\\ \hline
		
		$\mathbf{X}_{[N]}$, $\mathbf{X}_{[M,N]}$ &
		\begin{tabular}{c} $N\times N$-dimensional matrix, \\ $M\times N$-dimensional matrix 
		\end{tabular}\\ \hline
		
		$\mathbf{I}$, $\mathbf{O}$ & Identity matrix, zero matrix \\ \hline
		
		\raisebox{-1pt}{$\mathbf{X}^\top$} & Transpose of a matrix $\mathbf{X}$ \\ \hline
		
		$x_n$ and $[\mathbf{x}]_n$ & $n$-th element of a vector $\mathbf{x}$ \\ \hline
		
		$x_{m,n}$ and $[\mathbf{X}]_{m,n}$ & $(m,n)$-th element of a matrix $\mathbf{X}$ \\ \hline
		
		\raisebox{-1pt}{$\mathrm{vec}(\mathbf{X})\in \mathbb{R}^{MN}$} &\raisebox{-1pt}{ Vectorization of $\mathbf{X}\in\mathbb{R}^{M \times N}$} \\ \hline
		
		$\mathrm{vec}^{-1}_{(M,N)}(\mathbf{x})\in\mathbb{R}^{M \times N}$ & Matricization of $\mathbf{x}\in \mathbb{R}^{MN}$ \\ \hline
		
		\begin{tabular}{l}\hspace{-0.5em}
			\scriptsize
			$\mathrm{diag}(a_1,\ldots, a_{N}) \in \mathbb{R}^{N \times N}$, \\
			\scriptsize\hspace{-1.75em}
			$\mathrm{diag}(\mathbf{A}_1,\ldots, \mathbf{A}_{N})$	$\in \mathbb{R}^{\sum M_i \times \sum N_i}$
		\end{tabular} &
		\begin{tabular}{c}
			Diagonal/Block-diagonal matrix \\
			\scriptsize
			$(\mathbf{A}_i \in \mathbb{R}^{M_i \times N_i})$
		\end{tabular} \\ \hline
		
		$\lambda_{\mathrm{max}}^{++}(\mathbf{X})$ & Largest eigenvalue of $\mathbf{X}$ \\ \hline
		
		$\sigma_n(\mathbf{X})$ & $n$-th largest singular value of $\mathbf{X}$ \\ \hline
		
		$\|\mathbf{X}\|_{\ast}$ &
		\begin{tabular}{c}
			Nuclear norm \\
			$\|\mathbf{X}\|_{\ast} = \sum_{n=1}^{\min\{M,N\}} \sigma_n(\mathbf{X})$
		\end{tabular} \\ \hline
		
		$\|\mathbf{x}\|_{p}\quad (p \in [1 , \infty ))$&
		\begin{tabular}{c}
			$\ell_p$-norm, \\
			$\|\mathbf{x}\|_{p} = \left(\sum^{N}_{n=1} |x_n|^p\right)^{\frac{1}{p}}$
		\end{tabular} \\ \hline		
		
		$\|\mathbf{X}\|_{p}\quad (p \in [1 , \infty ))$&
		\begin{tabular}{c}
			$\ell_p$-matrix norm \\
			\scriptsize
			\hspace{-1em}
			$\|\mathbf{X}\|_{p} = \left(\sum^{M}_{m=1}\sum^{N}_{n=1} |x_{m,n}|^p\right)^{\frac{1}{p}}$
		\end{tabular} \\ \hline
		
		$\|\mathbf{x}\|_{p,q}^{(N)}\quad (p,\ q \in [1 , \infty ))$ &
		\begin{tabular}{c}
			$\ell_{p,q}$-mixed norm \\
			\hspace{-0.5em}\scriptsize
			$\mathbf{x} = \begin{bmatrix}
				\mathbf{x}_1^\top & \mathbf{x}_2^\top & \cdots & \mathbf{x}_{K}^\top 
			\end{bmatrix}^\top \hspace{-0.5em}\in \mathbb{R}^{KN}$ \\
			$\|\mathbf{x}\|_{p,q}^{(N)}=\left(\sum^{K}_{k=1}\|\mathbf{x}_k\|_p^q\right)^{\frac{1}{q}}$
		\end{tabular} \\ \hline
				
		$\mathfrak{B}_p(\y,\,\varepsilon) \quad (p,\ q \in [1 , \infty ))$ & 
		\begin{tabular}{c}
			$\ell_{p}$-ball centered at $\y$ \\
			$\{\x\ |\ \|\x - \y\|_p\leq\varepsilon\}$
		\end{tabular} \\ \hline
		
		$\mathbf{x} \leq \mathbf{y}\ (\mathbf{x},\ \mathbf{y}\in \mathbb{R}^{N})$ & $x_i \leq y_i, (1 \leq \forall i \leq N)$ \\ 
		\hline
		
		$\mathbf{x} \lneqq_{\mathcal{I}} \mathbf{y}\ (\mathbf{x},\ \mathbf{y}\in \mathbb{R}^{N},N\geq2)$ & 
		\begin{tabular}{c}
			$x_i < y_i, (i \in \mathcal{I}\subset\{1,...,N\})$
			\\
			$x_i = y_i, (i \in \{1,...,N\}\backslash\mathcal{I}) $
		\end{tabular}
		\\ 
		\thline
	\end{tabular}
\end{table}
\endgroup

This section reviews the basic tools for convex optimization used in this paper (for more detailed information, see \cite{Bauschke2011}). Mathematical notations are summarized in Table \ref{tab:notations}. Boldfaced large and small letters are matrices and vectors, respectively.

\subsection{Convex Scalar/Vector-valued Function}
A scalar/vector-valued function $f:\mathbb{R}^{N} \rightarrow \mathbb{R}^M$ is said to be a convex function \cite{Boyd2004} if and only if, for $\forall \mathbf{x},\ \mathbf{y} \in \mathbb{R}^{N}$ and $\forall \alpha \in [0,1]$, 
\begin{align}\label{eq:conv_vecconv}
	f(\alpha \mathbf{x} + (1-\alpha)\mathbf{y}) \leq \alpha f(\mathbf{x}) + (1-\alpha)f(\mathbf{y}).
\end{align}
As a special class of vector functions, for a vector with non-overlapped $K$ blocks $\mathbf{x} = \begin{bmatrix}
	\mathbf{x}_1^\top & \ldots & \mathbf{x}_{K}^\top \end{bmatrix}^\top \in \mathbb{R}^{N}$ ($\mathbf{x}_k \in \mathbb{R}^{N_k}$, $1 \leq k \leq K$, $N = \sum_{k=1}^{K} N_k$), we define \textit{block-wise vector function} $f:\mathbb{R}^{N} \rightarrow \mathbb{R}^K$ with functions for each block $f_k:\mathbb{R}^{N_k} \rightarrow \mathbb{R}$ as $
	f(\mathbf{x}) =  \begin{bmatrix}
		f_1(\mathbf{x}_1) & \ldots & f_{K}(\mathbf{x}_{K})
	\end{bmatrix}^\top.$
If a block-wise vector function satisfies convexity, we refer to it as \textit{convex block-wise vector function}. 

\subsection{Epigraph}\label{subsec:epi}
An epigraph\cite{Beck2017} of a function $f:\mathbb{R}^N\rightarrow\mathbb{R}^{M}$ is a subset of the product space $\mathbb{R}^N \times \mathbb{R}^{M}$ defined by 
\begin{align}
	\mathrm{epi}_f := \{ (\mathbf{x},\,\boldsymbol{\xi}) \in \mathbb{R}^N \times \mathbb{R}^{M}\  |\ f(\mathbf{x}) \leq \boldsymbol{\xi} \}. 
\end{align}
If $f \in\Gamma_0(\mathbb{R}^N)$ (the set of proper lower semicontinuous convex functions \cite{Bauschke2011}), then $\mathrm{epi}_f$ is a convex set.

\subsection{Proximity Operator and Projection onto Convex Set}
The proximity operator for a function $f:\mathbb{E}\subset \mathbb{R}^N\rightarrow \mathbb{R}$ and an index $\gamma \in \R_{++}$, denoted as $\mathrm{prox}_{\gamma f}:\mathbb{E} \rightarrow \mathbb{E}$, is defined by
\begin{align}\label{eq:Def_prox}
	\mathrm{prox}_{\gamma f}(\mathbf{x}) := \argmin_{\mathbf{y}\in \mathbb{E}} \gamma f(\mathbf{y}) + \frac{1}{2} \|\mathbf{x}-\mathbf{y}\|^2_2.
\end{align}
Projection onto a convex set $C \subset \mathbb{R}^N$, denoted as $\mathcal{P}_C(\mathbf{x}):\mathbb{R}^{N} \rightarrow \mathbb{R}^{N}$, is one of the cases of the proximity operator, derived as
$
\mathrm{prox}_{\gamma \iota_C}(\mathbf{x}) := \argmin_{\mathbf{y}\in \mathbb{R}^N} \gamma \iota_C(\mathbf{y}) + \frac{1}{2} \|\mathbf{x}-\mathbf{y}\|^2_2 = \argmin_{\mathbf{y} \in C} \|\mathbf{x}-\mathbf{y}\|_2^2 =: \mathcal{P}_C(\mathbf{x}),
$
where $\iota_C$ is the indicator function defined by $\iota_{C}(\mathbf{x}) = 0\ (\mathbf{x} \in C)$ and $\iota_{C}(\mathbf{x}) = \infty\ (\mathbf{x} \notin C)$.

Examples of the proximity operators are listed
in Table \ref{tab:complexity}. Although we assume the convexity for functions in general, similar to how the proximity operator for the $\ell_0$ pseudo-norm can be computed using the hard-thresholding operator\cite{Donoho1994}, the proximity operators for many non-convex functions can also be computed in closed-form\cite{Woodworth2016,Marjanovic2012,Xu2012}.

\begingroup 
\begin{table}[t]
	\caption{Proximity operator and computational complexity for $N$- and $M \times N$-dimensional vectors/matrices.}
	\label{tab:complexity}
	\centering
	{
		\footnotesize
		\begin{tabular}{>{\centering\arraybackslash}p{5.5em}|>{\arraybackslash\hspace{-0.5em}}p{15.5em}>{\arraybackslash}p{6em}}
			\thline
			Function $f$ & \begin{tabular}{c} Proximity operator $\widetilde{\mathbf{x}} = \mathrm{prox}_{\gamma f}(\mathbf{x})$\end{tabular} & \scriptsize \begin{tabular}{c}Computational\\complexity\end{tabular} \\ \thline
			$\|\cdot\|_{2}$  & \begin{tabular}{l} $\widetilde{\mathbf{x}} = \left(1 - \frac{\gamma}{\max\{\|\mathbf{x}\|_2,\gamma\}}\right)\mathbf{x}$ \cite{Beck2017}. \end{tabular}& \scriptsize\hspace{1em}$\mathcal{O}(N).$ \\\hline
			$\|\cdot\|_{1}$  & \begin{tabular}{l} Soft-thresholding $\widetilde{\mathbf{x}} = \mathcal{T}_\gamma(\mathbf{x})$ \cite{Beck2017} \\ \scriptsize\hspace{-0.25em}$\left( [\mathcal{T}_\gamma(\mathbf{x})]_n = \mathrm{sign}(x_n)\max\{|x_n|-\gamma, 0\}\right) $. \end{tabular}& \scriptsize\hspace{1em}$\mathcal{O}(N).$\\\hline
			$\iota_{\mathfrak{B}_1(\y,\,\varepsilon)}(\cdot)$ & \begin{tabular}{l} Using projection onto the simplex \cite{Duchi2008}.
			\end{tabular}& \scriptsize$\hspace{1em}\mathcal{O}(N\log N).$\\\hline
			$\|\cdot\|_{\infty}$  & \begin{tabular}{l} $\widetilde{\mathbf{x}} = \mathbf{x}-\mathrm{prox}_{\iota_{\mathcal{B}_1(\mathbf{0},1)}}(\mathbf{x}/\gamma)$ \cite{Beck2017}.
			\end{tabular} & \scriptsize$\hspace{1em}\mathcal{O}(N\log N).$\\\hline
			\begin{tabular}{l}  $\|\mathbf{x}\|_{2,1}$ \end{tabular} &  \begin{tabular}{l} Block-wise proximity operator of $\|\cdot\|_2$ \\ $( \mathbf{x} = \begin{bmatrix} \mathbf{x}_1^\top & \cdots & \mathbf{x}_{N}^\top \end{bmatrix}^\top,$\\$\widetilde{\mathbf{x}}_n = \mathrm{prox}_{\gamma \|\cdot\|_2} (\mathbf{x}_n))$.  \end{tabular} & \scriptsize\hspace{1em}$\mathcal{O}(N).$
			\\\hline
			$\|\mathbf{X}\|_{\ast}$  & \begin{tabular}{l} Soft-thresholding on singular values \\ of $\mathbf{X}\in\R^{M\times N}$ \cite{Beck2017} \\\scriptsize\hspace{-0.25em}$( \mathbf{X}=\mathbf{U}\mathrm{diag}({\bm \sigma})\mathbf{V}^\top,\ {\bm \sigma} = \begin{bmatrix} \sigma_1\ \ldots \ \sigma_r \end{bmatrix}^\top\hspace{-1em}, $\\\scriptsize $ \widetilde{\mathbf{X}} = \mathbf{U}\mathrm{diag}(\mathcal{T}_\gamma({\bm \sigma}))\mathbf{V}^\top) $.
			\end{tabular} & \scriptsize\begin{tabular}{l}
			$\mathcal{O}(MN $\\$\min\{M,N\}).$
			\end{tabular} \\\hline
			$\iota_{[a,b]^N}(\cdot)$ & \begin{tabular}{l}
				$\widetilde{\x}_n = \min\{ \max\{ \x_n , a \} , b \}$ \cite{Beck2017}
			\end{tabular}.&\scriptsize\hspace{1em}$\mathcal{O}(N)$.\\\hline
			$\hspace{-0.5em}\iota_{\mathrm{epi}_{\tau\|\cdot\|_2}}(\mathbf{x},\xi)$ & \begin{tabular}{l}\scriptsize\hspace{-0.25em}$
			  	(\widetilde{\mathbf{x}},\widetilde{\xi})
			  	=\begin{cases}
			  		(\mathbf{x},\xi) & (\textcolor{black}{\tau\|\mathbf{x}\|_2 \leq \xi})\\
			  		(\mathbf{0},0) & (\|\mathbf{x}\|_2 < -\tau\xi)\\
			  		\alpha\left(\mathbf{x},\tau\|\mathbf{x}\|_2\right) & (\mathrm{otherwise})
			  	\end{cases}
			  $, \\ $\alpha = \frac{1}{1+\tau^2}\left(1+\frac{\tau \xi}{\|\mathbf{x}\|_2}\right)$ \cite{Beck2017}.
			  \end{tabular}&\scriptsize\hspace{1em}$\mathcal{O}(N).$\\\hline
			$\hspace{-0.25em}\iota_{\mathrm{epi}_{\|\cdot\|_1}}(\mathbf{x},\xi)$ & \begin{tabular}{l}\scriptsize\hspace{-0.25em}$
			  (\widetilde{\mathbf{x}},\widetilde{\xi}) =\begin{cases}
			  	(\mathbf{x},\xi) & (\|\mathbf{x}\|_1 \leq \xi) \\
			  	(\mathcal{T}_{\gamma^\star}(\mathbf{x}),\xi+\gamma^\star) & (\|\mathbf{x}\|_1 > \xi) 
			  \end{cases}
			  $, \\ $\gamma^\star$ is any positive root of \\ $\varphi(\gamma) :=\|\mathcal{T}_{\gamma}(\mathbf{x})\|_1 - \gamma - \xi\ $ \cite{Beck2017}.
			  \end{tabular} &  \scriptsize\hspace{1em}$\mathcal{O}(N\mathrm{log}N)$.\\\hline
			$\hspace{-0.5em}\iota_{\mathrm{epi}_{\|\cdot\|_\ast}}(\mathbf{X},\xi)$ & \begin{tabular}{l} Using projection onto the $\mathrm{epi}_{\|\cdot\|_1}$ \\ on singular values of $\mathbf{X}\in\R^{M\times N}$ \cite{Beck2017} \\ \scriptsize
			( $\mathbf{X}=\mathbf{U}\mathrm{diag}({\bm \sigma})\mathbf{V}^\top$,\\\scriptsize $(\mathbf{\bm \sigma}^\star,\xi^\star) = \mathrm{prox}_{\iota_{\mathrm{epi}_{\|\cdot\|_{1}}}}(\mathbf{\bm \sigma},\xi)$, \\\scriptsize
			$(\widetilde{\mathbf{X}},\widetilde{\xi})=(\mathbf{U}\mathrm{diag}({\bm \sigma}^\star)\mathbf{V}^\top , \xi^\star)$ ).
			\end{tabular}& \scriptsize\begin{tabular}{l}
			$\mathcal{O}(MN $\\$\min\{M,N\}).$
			\end{tabular} \\\thline
		\end{tabular}
	}
	\vspace{-0.4cm}
\end{table}
\endgroup 

\subsection{Constrained Linearly-Involved Generalized Moreau-Enhanced Model}
\label{sec:LiGME}
Among CPNRs, we adopt the constrained LiGME (cLiGME) model\cite{Yata2022} in our method, owing to its convexity-preserving property and applicability for optimization with epigraph constraints. This section reviews the cLiGME model. For a convex regularizer (seed function) $\Psi \in \Gamma_0(\mathbb{R}^{N_{\boldsymbol{\mathfrak{L}}}})$ and $\boldsymbol{\mathfrak{L}} \in \mathbb{R}^{N_{\boldsymbol{\mathfrak{L}}}\times N}$, the constrained nonconvexly-regularized convex optimization with the LiGME penalty $\Psi_{\mathbf{B}}$ is formulated as follows:
\begin{align}
	&\argmin_{\mathbf{x} \in \mathbb{R}^N,\mathbf{Cx} \in \mathcal{C}_0}\frac{1}{2}\|\boldsymbol{\Phi}\mathbf{x} -\mathbf{y}\|_2^2 + \mu {\Psi}_{\mathbf{B}}(\boldsymbol{\mathfrak{L}}\mathbf{x}),\nonumber\\
	=&\,\argmin_{\mathbf{x} \in \mathbb{R}^N}\frac{1}{2}\|\boldsymbol{\Phi}\mathbf{x} -\mathbf{y}\|_2^2 + \mu \Psi_{\mathbf{B}}(\boldsymbol{\mathfrak{L}}\mathbf{x}) + \iota_{\mathcal{C}_0}(\mathbf{C}\mathbf{x}),\label{eq:cLiGMEreg}\\
	&{\Psi}_{\mathbf{B}}(\cdot) := {\Psi}(\cdot) - \min_{\mathbf{v}}\left\{ {\Psi}(\mathbf{v}) + \frac{1}{2}\|\mathbf{B}(\cdot - \mathbf{v})\|_2^2\right\},\nonumber
\end{align}
where $\mathbf{B} \in \mathbb{R}^{N_\B\times N_{\boldsymbol{\mathfrak{L}}}}$ controls the strength of non-convexity of the LiGME penalty, $\mathbf{C} \in \mathbb{R}^{N_\mathbf{C} \times N}$, and $\mathcal{C}_0 \subset \mathbb{R}^{N_\mathbf{C}}$. 
To preserve the overall convexity property of \eqref{eq:cLiGMEreg}, the matrices $\boldsymbol{\Phi}$,  $\mathbf{B}$, and $\boldsymbol{\mathfrak{L}}$ should satisfy the following positive semi-definiteness:
\begin{align}\label{eq:convcond}
	\boldsymbol{\Phi}^\top\boldsymbol{\Phi} - \mu \boldsymbol{\mathfrak{L}}^\top\mathbf{B}^\top\mathbf{B}\boldsymbol{\mathfrak{L}} \succeq \mathbf{O}.
\end{align}
The solver of \eqref{eq:cLiGMEreg} is summarized in Algorithm \ref{alg:cLiGME}, where $\mathrm{prox}_{f\oplus\iota_{\mathcal{C}_0}}(\mathbf{w}^{(1)}, \mathbf{w}^{(2)}) = (\mathrm{prox}_{f}(\mathbf{w}^{(1)}), \mathcal{P}_{\mathcal{C}_0}(\mathbf{w}^{(2)}))$, and the parameters $\sigma$ and $\tau$ should satisfy the certain condition in \cite{Yata2022}.

\begin{algorithm}[t]
	\caption{Solver of cLiGME model \eqref{eq:cLiGMEreg}.}
	\label{alg:cLiGME}
	{\footnotesize
	\begin{algorithmic}[1]
		\STATE Define $\mathbf{B}$ satisfying the condition in \eqref{eq:convcond}.
		\STATE Set $k=0$ and choose initial parameters $\mathbf{x}_0$,  $\mathbf{v}_0$,  $\mathbf{w}_0$.
		\STATE Set ${\mathbf{p}}_{0} = \begin{bmatrix}
			\mathbf{x}_{0}^\top & \mathbf{v}_{0}^\top & \mathbf{w}_{0}^\top
		\end{bmatrix}^\top.$
		\STATE \algorithmicdo
		\STATE $\quad\widetilde{\mathbf{x}}_{k} \leftarrow  \left[ \mathbf{I} - \frac{1}{\sigma}\left(\boldsymbol{\Phi}^\top\boldsymbol{\Phi}-\mu \boldsymbol{\mathfrak{L}}^\top\mathbf{B}^\top\mathbf{B}\boldsymbol{\mathfrak{L}}\right)\right]\mathbf{x}_k$
		\STATE $\quad{\mathbf{x}}_{k+1} \leftarrow  \widetilde{\mathbf{x}}_{k} - \frac{\mu}{\sigma}\boldsymbol{\mathfrak{L}}^\top\mathbf{B}^\top\mathbf{B}\mathbf{v}_k -  \frac{\mu}{\sigma}\begin{bmatrix}
			\boldsymbol{\mathfrak{L}}^\top & \C^\top
		\end{bmatrix}\mathbf{w}_k + \frac{1}{\sigma}\boldsymbol{\Phi}^\top\mathbf{y}$
		\STATE  $\quad\widetilde{\mathbf{v}_{k}}\leftarrow \left(\mathbf{I} - \frac{\mu}{\tau}\mathbf{B}^\top\mathbf{B} \right) \mathbf{v}_{k}$
		\STATE  $\quad{\mathbf{v}}_{k+1} \leftarrow  \mathrm{prox}_{\frac{\mu}{\tau}{\Psi}}\left[ \frac{2\mu}{\tau}\mathbf{B}^\top\mathbf{B}\boldsymbol{\mathfrak{L}}\mathbf{x}_{k+1} - \frac{\mu}{\tau}\mathbf{B}^\top\mathbf{B}\boldsymbol{\mathfrak{L}}\mathbf{x}_{k} - \widetilde{\mathbf{v}_{k}}\right]$
		\STATE $\quad{\mathbf{w}}_{k+1} \leftarrow  \left(\mathbf{I}-\mathrm{prox}_{{\Psi}\oplus\iota_{\mathcal{C}_0}}\right)\left[ \begin{bmatrix}
		\boldsymbol{\mathfrak{L}}^\top & \C^\top
		\end{bmatrix}^\top
		(2\mathbf{x}_{k+1} - \mathbf{x}_{k}) + \mathbf{w}_{k}\right]$
		\STATE $\quad{\mathbf{p}}_{k+1} \leftarrow  \begin{bmatrix}
			\mathbf{x}_{k+1}^\top & \mathbf{v}_{k+1}^\top & \mathbf{w}_{k+1}^\top
		\end{bmatrix}^\top$
		\STATE $\quad k \leftarrow  k+1$
		\STATE \algorithmicwhile$\ \|{\mathbf{p}}_{k} - {\mathbf{p}}_{k-1}\|_2 > \varepsilon_{\mathrm{stop}}$
		\STATE Output $\mathbf{x}_{k}$.
	\end{algorithmic}
	}
\end{algorithm}

\section{Multilayered Non-convex Regularization with Epigraphical Relaxation for Signal Recovery}\label{sec:ERNCR}
In this section, we propose a signal recovery method enhanced by multilayered non-convex regularization with ER. Firstly, Sec. \ref{sec:basicSRform} formulates a general multilayered regularization problem. In order to handle it by the existing optimization algorithms, Sec. \ref{sec:conventionalER} extends the ER technique for multilayered convex regularization problems to the non-convex one with theoretical guarantee of the equivalence between the global optimal solution sets with and without the ER, under a milder condition compared to \cite{Kyochi2021}.

\subsection{Problem Formulation for Signal Recovery}\label{sec:basicSRform}
Signal recovery can be formulated as an inverse problem aiming to estimate a signal of interest $\mathbf{x} \in \mathbb{R}^N$ from noisy observation $\mathbf{y}\in\R^M$ modeled as
\begin{align}\label{eq:observation}
	\y = \boldsymbol{\Phi}\x + \n,
\end{align}
where $\boldsymbol{\Phi} \in \mathbb{R}^{M\times N}$ is the degradation process (e.g., blurring, downsampling, and so on) and $\mathbf{n} \in \mathbb{R}^M$ is the additive noise, e.g., Gaussian/sparse noise. In certain signal recovery scenarios, such as compressed sensing, the inverse problem tends to be ill-posed ($M<N$). To achieve robust estimation of $\x$ by incorporating its prior information, e.g., group-sparsity/low-rankness, we translate it into the minimization problem including a $K$-layered composite (possibly non-convex) regularization function $\mathcal{R}^{(1,K)} = \mathcal{R}^{(K)} \circ \mathcal{R}^{(K-1)} \circ \cdots \circ \mathcal{R}^{(1)}: \mathbb{R}^{N_1} \rightarrow \mathbb{R}_{++}$ ($\mathcal{R}^{(k)}:\mathbb{R}^{N_k} \rightarrow \mathbb{R}_{++}^{N_{k+1}}$, $1 \leq k \leq K,\ \mathbb{R}_{++}^{N_{K+1}} = \mathbb{R}_{++})$ with the regularization parameter $\mu\in\mathbb{R}_{++}$ as
\begin{align}\label{eq:genform}
	\mathcal{S}_{\mathbf{x}} =\argmin_{\mathbf{x}\in \mathcal{C} \subset \mathbb{R}^{N}} \mathcal{U}(\x) \left( := \mathcal{F}_\y(\mathbf{\Phi}\mathbf{x}) + \mu \mathcal{R}^{(1,K)}(\mathbf{A}\mathbf{x})\right),
\end{align}
where $\mathcal{U}:\mathcal{C} \rightarrow \mathbb{R}_{++}$, $\mathbf{A}\in \mathbb{R}^{N_1 \times N}$, and  $\mathcal{F}_\mathbf{y}\in\Gamma_0(\R^{M})$ is a coercive function (i.e., if $\|\x\|\rightarrow\infty$, then $\mathcal{F}_\y(\x)\rightarrow\infty$) consisting of the sum of proximable functions $\mathcal{F}_\y(\mathbf{\Phi}\x):=\sum^Q_{q=1} \mathcal{F}_{\y}^{(q)}(\mathbf{\Phi}_q \x)\, (\mathcal{F}_{\y}^{(q)}:\R^{M_q}\rightarrow \R_{++},\, \mathbf{\Phi}_q\in\R^{M_q\times N},\, M=\sum^Q_{q=1} M_q$), e.g., $\mathcal{F}_\y(\mathbf{\Phi}\mathbf{x})=\tfrac{1}{2}\|\mathbf{\Phi}\mathbf{x}-\y\|_2^2$. An example of $\mathcal{R}^{(1,K)}$ is the $\ell_{2,1}$-norm employed in TV regularization (see Sec. \ref{sec:problemERLiGME}), where $\|\cdot\|_{2,1} = \mathcal{R}^{(2)} \circ \mathcal{R}^{(1)}$, $\mathcal{R}^{(1)}$ is a {\it{block-wise}} $\ell_2$-norm and $\mathcal{R}^{(2)} = \|\cdot\|_1$.

Proximal splitting algorithms, e.g., the alternating direction method of multipliers (ADMM), and primal-dual splitting (PDS) methods\cite{Gabay1976,Beck2009,Chambolle2010,Condat2013,Vu2013}, are mainstream strategies for solving the convex optimization problem in \eqref{eq:genform} with a convex regularizer, if the closed-form proximity operator for $\mathcal{R}^{(1,K)}$ is available. Even for a non-convex one, various studies, such as \cite{Wang2019} and \cite{Bot2020}, have proposed the algorithms based on the proximity operator for finding stationary (or Karush-Kuhn-Tucker (KKT)) points under certain settings and conditions.

\subsection{Epigraphical Relaxation for Non-convex Multilayered Regularizer}\label{sec:conventionalER}
It is often the case that the proximal operators of multilayered non-convex regularization functions cannot be computed in closed-form (or can be computed but not trivially). To address this issue, we extend ER\cite{Kyochi2021}, which was originally proposed for convex multilayered regularization problems, to non-convex ones. ER enables us to handle the multilayered regularization problem by reformulating it into a relaxed problem as follows.
First, we recursively introduce an auxiliary variables $\{\mathbf{z}^{(K)},...,\mathbf{z}^{(2)}\}$ (denoted as $\{\mathbf{z}^{(k)}\}^2_{k=K}$) into the inner functions of $\mathcal{R}^{(1,K)}(\mathbf{A}\mathbf{x})$ in Problem \eqref{eq:genform} as $\mathbf{z}^{(K)} = \mathcal{R}^{(K-1)}(\mathbf{z}^{(K-1)}),\,\mathbf{z}^{(K-1)} = \mathcal{R}^{(K-2)}(\mathbf{z}^{(K-2)}),...,\mathbf{z}^{(2)} = \mathcal{R}^{(1)}(\mathbf{A}\mathbf{x})$ and reformulate the problem as
\begin{align}\label{eq:reform2}
	&\mathcal{S}_{\mathbf{x}}\times\prod_{k=K}^{2}\mathcal{S}_{\mathbf{z}^{(k)}}  := \argmin_{( \mathbf{x}, \{\mathbf{z}^{(k)}\}_{k=K}^{2}) \in \mathcal{D} } \mathcal{V}\left( \mathbf{x}, \{\mathbf{z}^{(k)}\}_{k=K}^{2}\right),\nonumber\\ 
	&\mathcal{V}( \mathbf{x}, \{\mathbf{z}^{(k)}\}_{k=K}^{2}):=\mathcal{F}_\y(\mathbf{\Phi x}) + \mu \mathcal{R}^{(K)}(\mathbf{z}^{(K)}), \nonumber\\ 
	& \mathcal{D} := \left\lbrace (\mathbf{x}, \{\mathbf{z}^{(k)}\}_{k=K}^{2})\,|\,\x\in\mathcal{C},\,\z^{(2)} = \mathcal{R}^{(1)}(\mathbf{A}\mathbf{x}), \right. \nonumber\\ &\left. \hspace{4em}\z^{(k+1)} =\mathcal{R}^{(k)}(\z^{(k)})\,(2\leq k \leq K-1) \right\rbrace,
\end{align}
where $\mathcal{V}:\mathcal{D}\rightarrow\R_{++}$. Subsequently, by relaxing all of the above equality constraints to the in inequality ones with the epigraph as $\z^{k+1}\geq \mathcal{R}^{(k)}(\z^{(k)}) \Leftrightarrow (\z^{k},\,\z^{k+1})\in \mathrm{epi}_{\mathcal{R}^{(k)}}$, we finally obtain the epigraphically-relaxed problem as follows:
\begin{align}\label{eq:reform4}
	&\widetilde{\mathcal{S}}_{\mathbf{x}}\times\prod_{k=K}^{2}\widetilde{\mathcal{S}}_{\mathbf{z}^{(k)}}  := \argmin_{( \mathbf{x}, \{\mathbf{z}^{(k)}\}_{k=K}^{2}) \in \widetilde{\mathcal{D}} } \widetilde{\mathcal{V}}\left( \mathbf{x}, \{\mathbf{z}^{(k)}\}_{k=K}^{2}\right),\nonumber\\ 
	&\widetilde{\mathcal{V}}( \mathbf{x}, \{\mathbf{z}^{(k)}\}_{k=K}^{2}):=\mathcal{F}_\y(\mathbf{\Phi x}) + \mu \mathcal{R}^{(K)}(\mathbf{z}^{(K)}), \nonumber\\ 
	& \widetilde{\mathcal{D}} := \left\lbrace (\mathbf{x}, \{\mathbf{z}^{(k)}\}_{k=K}^{2})\,|\,\x\in\mathcal{C},\,(\mathbf{A}\mathbf{x},\mathbf{z}^{(2)}) \in \mathrm{epi}_{\mathcal{R}^{(1)} }, \right. \nonumber\\ &\left. \hspace{4em}(\mathbf{z}^{(k)},\mathbf{z}^{(k+1)}) \in \mathrm{epi}_{\mathcal{R}^{(k)} } \ (2 \leq k \leq K-1)) \right\rbrace,
\end{align}
where $\widetilde{\mathcal{V}}:\widetilde{\mathcal{D}}\rightarrow\R_{++}$. As long as both the proximity operator for the outermost regulatizer $\mathcal{R}^{(K)}$ and the projection onto each epigraphical constraint are available, we can tackle the relaxed problem by proximal splitting algorithms.

Although the relaxed problem in \eqref{eq:reform4} is not the same problem as \eqref{eq:genform}, it has been proven that the set of the global minimizers of the relaxed problem is equivalent to that of the original problem in \eqref{eq:genform}, under the assumptions \textcolor{black}{(A1) and (A2)}, in \cite{Kyochi2021}.
\begin{itemize}
	\item[(A1)] 
	The $k$-th layer function $\mathcal{R}^{(k)}$ $(1 \leq \forall k \leq K)$ is a convex function ($k=K$) or convex block-wise vector functions ($1\leq k \leq K-1$).
	\item[(A2)] 
	The $k$-th layer function $\mathcal{R}^{(k)}$ ($2 \leq \forall k \leq K$) is \textit{strictly increasing (scalar/vector) function}\footnote{A function $f:C \rightarrow \mathbb{R}$ ($C \subset \mathbb{R}^N$) is said to be a strictly increasing function on $C$ if for any $\mathbf{x} , \mathbf{y} \in C$ satisfying $\mathbf{x} \leq \mathbf{y}$ and $x_{n} < y_{n}$ for some $1\leq n \leq N$, then $f(\mathbf{x}) < f(\mathbf{y})$. A block-wise vector function $f:\mathbb{R}^N \rightarrow \mathbb{R}^M$
	is said to be a strictly increasing block-wise vector function if all of the block functions $f_{m} : \mathbb{R}^{N_m} \rightarrow \mathbb{R}$ ($1\leq m\leq M$) are strictly increasing.} on $\mathbb{R}_{++}^{N_{k}}$.
\end{itemize}

On the other hand, in this paper, we analyze the case that $\mathcal{R}^{(1,K)}$ is a non-convex composite function, i.e., it satisfies only the assumption (A2), and derive the equivalence of the solution sets in \eqref{eq:genform} and \eqref{eq:reform4} without the assumption (A1) as in Proposition \ref{Theorem1}. Proposition \ref{Theorem1} is proved using the following lemma.

\begin{lemma}\label{Lemma1}
	$\forall( {\mathbf{x}}^\star, \{{\mathbf{z}}^{\star(k)}\}_{k=K}^{2})\in\widetilde{\mathcal{S}}_{\mathbf{x}}\times\prod_{k=K}^{2}\widetilde{\mathcal{S}}_{\mathbf{z}^{(k)}}\,(\subset\widetilde{\mathcal{D}})$ in \eqref{eq:reform4} satisfies the following equations:
	\begin{align}\label{eq:lemma1}
		\mathcal{R}^{(1)}(\mathbf{A}{\mathbf{x}}^\star) = {\mathbf{z}}^{\star(\textcolor{black}{2})},\quad\mathcal{R}^{(k)}({\mathbf{z}}^{\star(k)})=\textcolor{black}{{\mathbf{z}}^{\star(k+1)},}
	\end{align}
	where $2\leq k\leq K-1$. Thus, $( {\mathbf{x}}^\star, \{{\mathbf{z}}^{\star(k)}\}_{k=K}^{2})\in{\mathcal{D}}$.
\end{lemma}
\begin{proof}
	see Appendix \ref{ap:poof_lemma1}.
\end{proof}

\begin{prop}[Equivalence of the solutions between Problem \eqref{eq:genform} and \eqref{eq:reform4}]\label{Theorem1}
	Let $\mathcal{S}_{\mathbf{x}}$ be the set of the global minimizers of a $K$-layered regularization problem in \eqref{eq:genform} where $\mathcal{R}^{(1,K)}$ is a non-convex (block-wise vector) function, and $\widetilde{\mathcal{S}}_{\mathbf{x}}$ be the set of the global minimizers of the corresponding epigraphically-relaxed problem in \eqref{eq:reform4}. If $\mathcal{R}^{(1,K)}$ satisfies the assumption (A2), then
	${\mathcal{S}}_{\mathbf{x}} =\widetilde{\mathcal{S}}_{\mathbf{x}} $.
\end{prop}
\begin{proof}
	see Appendix \ref{ap:poof_ER}.
\end{proof}

\section{Epigraphically-relaxed Linearly-involved generalized Moreau-Enhansed model}
Even though a solution can be found in non-proximable multilayered non-convex regularization by using ER, it may still converge to a local optimal solution depending on initial values. To address this issue, this section proposes an epigraphically-relaxed LiGME model (ER-LiGME model) guaranteeing global optimal solutions for non-proximable multilayered non-convex regularization. Firstly, Sec. \ref{sec:ER}-B formulate multilayered regularization and discuss the problem that ER cannot be applied to conventional LiGME models straightforwardly. In Sec. \ref{subsec:ERLiGME}-E, we propose a new LiGME model that seamlessly integrates ER and LiGME, and discuss its setting and computational complexity. Sec. \ref{sec:ERLiGME_examples} shows effective non-proximable multilayered non-convex regularizers arising from our ER-LiGME model.

\subsection{Problem Formulation}\label{sec:ER}
Let us consider the $K$-layered composite regularization problem in \eqref{eq:genform} under the setting of $\mathcal{F}_\y(\boldsymbol{\Phi}\mathbf{x}) = \tfrac{1}{2}\|\boldsymbol{\Phi}\mathbf{x}-\mathbf{y}\|_2^2$, where $\mathcal{R}^{(1,K)}$ is a possibly non-proximable regularizer with
a proximable sparsity-promoting outermost convex function $\mathcal{R}^{(K)}$, and $\iota_{\mathrm{epi}_{\mathcal{R}^{(k)}}}$ ($2\leq \forall k\leq K-1$) are proximable. By using the expanded variable $\widehat{\mathbf{x}} := \begin{bmatrix}
	\mathbf{x}^{\top} & \mathbf{z}^{(K)\top} & \cdots & \mathbf{z}^{(2)\top}
\end{bmatrix}^{\top}\in \mathbb{R}^{N + \sum_{k=K}^2 N_k}$, the expanded degradation process $\boldsymbol{\Phi}_{\mathbf{O}} := \begin{bmatrix}
	\boldsymbol{\Phi}  & \mathbf{O} & \cdots &\mathbf{O}
\end{bmatrix} \in \mathbb{R}^{M\times (N + \sum_{k=K}^2 N_k)}$, the outermost regularizer $\widehat{\mathcal{R}}^{(K)}(\widehat{\mathbf{x}}) := {\mathcal{R}}^{(K)}(\mathbf{z}^{(K)})$, and indicator functions for epigraph constraints, its ER problem in \eqref{eq:reform4} can be rewritten as follows:
\begin{align}\label{eq:P1ER_reform}
	& \hspace{-0.75em}\argmin_{\widehat{\mathbf{x}}} 
	\frac{1}{2}\|\boldsymbol{\Phi}_{\mathbf{O}}\widehat{\mathbf{x}} -\mathbf{y}\|_2^2 
	+ \mu \widehat{\mathcal{R}}^{(K)}(\widehat{\mathbf{x}}) + \iota_{\mathrm{epi}_{\{\mathcal{R}^{(k)}\}_{k=K-1}^{1}}\cap \mathcal{C} }(\mathbf{C}\widehat{\mathbf{x}}), \nonumber \\
	& \iota_{\mathrm{epi}_{\{\mathcal{R}^{(k)}\}_{k=K-1}^{1}}\cap \mathcal{C} }(\mathbf{C}\widehat{\mathbf{x}})= \sum_{k=K-1}^{1} \iota_{\mathrm{epi}_{\mathcal{R}^{(k)}}}(\widehat{\mathbf{z}}_k) + \iota_{\mathcal{C} }(\mathbf{x}),
\end{align}
where the matrix $\mathbf{C}$ maps $\widehat{\mathbf{x}}$ to $\mathbf{C}\widehat{\mathbf{x}} =\begin{bmatrix}
	\widehat{\mathbf{z}}_{K-1}^\top & \cdots & \widehat{\mathbf{z}}_{1}^\top & \mathbf{x}^\top
\end{bmatrix}^\top$ with
$
\widehat{\mathbf{z}}_1= \begin{bmatrix} (\mathbf{A}\mathbf{x})^\top & \mathbf{z}^{(2)\top}\end{bmatrix}^\top$ and $
\widehat{\mathbf{z}}_k=\begin{bmatrix} \mathbf{z}^{(k)\top} & \mathbf{z}^{(k+1)\top}\end{bmatrix}^\top\ (k\geq 2).
$ In this paper, we apply the cLiGME model (reviewed in Sec. \ref{sec:LiGME}) to ${\mathcal{R}}^{(K)}$ for achieving sparsity-enhanced non-convex multilayered regularization with preserving overall convexity.

\subsection{Problem on Multilayered Regularization with cLiGME Model}
\label{sec:problemERLiGME}
Applying the cLiGME model to \eqref{eq:genform} as ${\Psi}=\mathcal{R}^{(1,K)}$ is usually difficult because multilayered regularizers tend to be non-proximable. ER has the potential to enhance $\mathcal{R}^{(1,K)}$ in the cLiGME model by relaxing the multilayered regularizer into the sum of proximable regularizer and indicator functions as in \eqref{eq:P1ER_reform}. However, straightforward integration of the cLiGME model and the ER technique disables non-convexification of the multilayered regularizer, and thus, yields a biased solution as explained in the following.

For simplicity, let us consider the 2-layered regularization problem employing the TV penalty defined as $\|\x\|_{\mathrm{TV}}=\|\mathbf{D}_{\mathrm{vh}}\mathbf{x}\|_{2,1}$ in \eqref{eq:genform}, where $\x\in\R^{N^2}$ is a vectorized $N\times N$ signal and $\mathbf{D}_{\mathrm{vh}}$ is the difference operator defined in Appendix \ref{subsec:EpiDSTV}\footnote{
	Since the $\ell_{2,1}$-norm is proximable, we can introduce the TV regularizer in the cLiGME model without ER as $\Psi_\B=\widehat{\Psi}_{\mathbf{B}_c}= (\|\cdot\|_{2,1})_\mathbf{B},\,\boldsymbol{\mathfrak{L}}=\boldsymbol{\mathfrak{L}}_c=\mathbf{D}_{\mathrm{vh}}$ by choosing a matrix $\mathbf{B}$ satisfying $\boldsymbol{\Phi}^\top\boldsymbol{\Phi} - \mu \mathbf{D}_{\mathrm{vh}}^\top\mathbf{B}^\top\mathbf{B}\mathbf{D}_{\mathrm{vh}} \succeq \mathbf{O}.$
}. When $\mathcal{R}^{(1,2)}(\x)=\|\x\|_{\mathrm{TV}}$, the ER problem in \eqref{eq:P1ER_reform} is formulated as follows:
\begin{align}\label{eq:problemTV}
	&\argmin_{\widehat{\mathbf{x}}\in\R^{2N^2}} \frac{1}{2}\| \boldsymbol{\Phi}_{\mathbf{O}}\widehat{\mathbf{x}} - \mathbf{y} \|_2^2 + \mu \widehat{\mathcal{R}}^{(2)}(\widehat{\mathbf{x}})  +  \iota_{\mathrm{epi}_{\|\cdot\|_2}}(\mathbf{C}\widehat{\mathbf{x}}),\nonumber\\
	&\boldsymbol{\Phi}_{\mathbf{O}} = \begin{bmatrix}
		\boldsymbol{\Phi} & \mathbf{O}_{[M,N^2]}
	\end{bmatrix},\ 
	\widehat{\mathbf{x}} = \begin{bmatrix}
		\mathbf{x}^\top & \mathbf{z}^\top
	\end{bmatrix}^\top,\ \widehat{\mathcal{R}}^{(2)}(\widehat{\mathbf{x}}) = \|\mathbf{z}\|_1,
\end{align}
where  ${\boldsymbol{\Phi}} \in \mathbb{R}^{M \times N^2}$, $\mathbf{y} \in \mathbb{R}^{M}$, and $\mathbf{C}\widehat{\mathbf{x}} = \begin{bmatrix} (\mathbf{D}_{\mathrm{vh}}\mathbf{x})^\top & \mathbf{z}^\top\end{bmatrix}^\top$. The straightforward integration of the ER problem in \eqref{eq:problemTV} into the cLiGME model with $\boldsymbol{\mathfrak{L}} = \mathbf{I}$ and $\mathbf{C} = \mathrm{diag}(
\mathbf{D}_{\mathrm{vh}} , \mathbf{I})$ is given as
\begin{align}\label{eq:problemTV2}
	&\argmin_{\widehat{\mathbf{x}}\in\R^{2N^2}} \frac{1}{2}\| \boldsymbol{\Phi}_{\mathbf{O}}\widehat{\mathbf{x}} - \mathbf{y} \|_2^2 + \mu \widehat{\mathcal{R}}^{(2)}_\mathbf{B}(\widehat{\mathbf{x}})  +  \iota_{\mathrm{epi}_{\|\cdot\|_2}}(\mathbf{C}\widehat{\mathbf{x}}).
\end{align}
Since $\boldsymbol{\Phi}_{\mathbf{O}}^\top\boldsymbol{\Phi}_{\mathbf{O}} = \mathrm{diag}({\boldsymbol{\Phi}}^\top{\boldsymbol{\Phi}}, \mathbf{O}_{[N^2]}) $, the matrix $\mathbf{B}$ should form, for example, $\mathbf{B}= \begin{bmatrix}
	\mathbf{B}_0 & \mathbf{O}_{[N_\B,N^2]}
\end{bmatrix}$, where  ${\boldsymbol{\Phi}}^\top{\boldsymbol{\Phi}} - \mu \mathbf{B}_0^\top\mathbf{B}_0 \succeq \mathbf{O}$, to satisfy the overall convexity condition \eqref{eq:convcond}. As the result, the LiGME penalty $\widehat{\mathcal{R}}^{(2)}_\mathbf{B}(\widehat{\mathbf{x}})$ is reformulated as
\vspace{-0.2cm}
\begin{align}\label{eq:problemTV3}
	\widehat{\mathcal{R}}^{(2)}_\mathbf{B}(\widehat{\mathbf{x}}) :=& \widehat{\mathcal{R}}^{(2)}(\widehat{\mathbf{x}}) - \hspace{-0.25em}\min_{\widehat{\mathbf{v}} = [ \mathbf{v}_1^\top \  \mathbf{v}_2^\top ]^\top } \hspace{-0.25em}\left\{ \widehat{\mathcal{R}}^{(2)}(\widehat{\mathbf{v}}) \hspace{-0.1em}+\hspace{-0.1em} \frac{1}{2}\|\mathbf{B}(\widehat{\mathbf{x}} - \widehat{\mathbf{v}})\|_2^2\right\} \nonumber\\
	=& \|\mathbf{z}\|_1 - \min_{\widehat{\mathbf{v}} = [ \mathbf{v}_1^\top \  \mathbf{v}_2^\top ]^\top } \left\{ \|\mathbf{v}_2\|_1 + \frac{1}{2}\|\mathbf{B}_0({\mathbf{x}} - {\mathbf{v}_1})\|_2^2\right\} \nonumber\\
	=&\|\mathbf{z}\|_1 \ \ (\because \mathbf{v}_1 = \mathbf{x},\ \mathbf{v}_2 = \mathbf{0}).
\end{align}
This equation \eqref{eq:problemTV3} shows that $\widehat{\mathcal{R}}^{(2)}_\mathbf{B}(\widehat{\mathbf{x}}) = \widehat{\mathcal{R}}^{(2)}(\widehat{\mathbf{x}})$ and thus the non-convexification of $\widehat{\mathcal{R}}^{(2)}$ is disabled. This result holds for any $K$-layered ($K\geq2$) regularization problem relaxed by ER with $\boldsymbol{\Phi}_{\mathbf{O}}=\begin{bmatrix}
	\boldsymbol{\Phi} & \mathbf{O} & ... & \mathbf{O}
\end{bmatrix}$.

\subsection{Epigraphically-Relaxed LiGME Model for Multilayered Regularization}\label{subsec:ERLiGME}
To solve the above problem, we propose the ER-LiGME model. It integrates the cLiGME model and the ER technique by a \textit{guided observation extension} (GOE), which appends a reference data $\mathbf{z}_{\mathrm{R}}\in\R^{N_K}$ for the variable ${\mathbf{z}}^{(K)}$ to the observation as $\widehat{\mathbf{y}} = \begin{bmatrix}
	\mathbf{y}^\top & \sqrt{\rho}\ \mathbf{z}_{\mathrm{R}}^{\top}
\end{bmatrix}^\top\,(\rho\in\R_{++})
$, and extends the degradation process $\boldsymbol{\Phi}_{\mathbf{O}}$ as follows:
\begin{align}\label{eq:prob1cELG}
	\argmin_{\widehat{\mathbf{x}}} &\frac{1}{2}\| \widehat{\boldsymbol{\Phi}}\widehat{\mathbf{x}} - \widehat{\mathbf{y}} \|_2^2 + \mu 
	\widehat{\mathcal{R}}_{\mathbf{B}}^{(K)}(\widehat{\mathbf{x}}) + \iota_{\mathrm{epi}_{\{\mathcal{R}^{(k)}\}_{k=K-1}^{1}}\cap \mathcal{C} }(\mathbf{C}\widehat{\mathbf{x}}), \nonumber\\
	&\widehat{\boldsymbol{\Phi}} =\ 
	\begin{bmatrix}
		\boldsymbol{\Phi} & \mathbf{O} & \mathbf{O} & \cdots & \mathbf{O}  \\ 
		\mathbf{O} & \sqrt{\rho}\mathbf{I}_{[N_K]} & \mathbf{O} & \cdots  & \mathbf{O} 
	\end{bmatrix}.
\end{align}
\begin{Rem}\label{rem:dataFidelty}
$ \frac{1}{2} \| \widehat{\boldsymbol{\Phi}}\widehat{\mathbf{x}} - \widehat{\mathbf{y}} \|_2^2$ can be decomposed as  $\frac{1}{2}\| {\boldsymbol{\Phi}}{\mathbf{x}} - {\mathbf{y}} \|_2^2 +  \frac{{\rho}}{2}\| \mathbf{z}^{(K)} - \mathbf{z}_{\mathrm{R}}\|_2^2
$, which implies that the $\mathbf{z}^{(K)}$ is optimized to be close to the reference data $\mathbf{z}_{\mathrm{R}}$.
\end{Rem}
This model satisfies the overall convexity condition \eqref{eq:convcond} by setting the matrix $\mathbf{B}$ with a non-convexification parameter $\theta\in\R_{++}$ as
\begin{align}\label{eq:settingB}
	\mathbf{B} := \mathrm{diag}(\mathbf{O}_{[N]}, \sqrt{\theta/\mu}\mathbf{I}_{[N_K]}, \mathbf{O}_{[N_{K-1}]}, \cdots, \mathbf{O}_{[N_1]}).
\end{align}
Recall that $\boldsymbol{\mathfrak{L}}=\mathbf{I}$ when ER is applied. Then, by simply setting the parameters to satisfy $0\leq \theta \leq \rho < \infty$, the overall convexity condition of this model is easily satisfied as
\begin{align}\label{eq:condGOE}
	\widehat{\boldsymbol{\Phi}}^\top\widehat{\boldsymbol{\Phi}} - \mu \mathbf{B}^\top\mathbf{B} &= \mathrm{diag}(\boldsymbol{\Phi}^\top\boldsymbol{\Phi}, (\rho-\theta)\mathbf{I}, \mathbf{O} , ...,\mathbf{O}) \succeq \mathbf{O}.
\end{align}
The condition for \eqref{eq:condGOE} is significantly simpler than \eqref{eq:convcond} as it does not involve the matrices $\boldsymbol{\Phi}$ and $\mathbf{A}$ (denoted as $\boldsymbol{\mathfrak{L}}$ in \eqref{eq:convcond}), but only depends on the parameters $\rho$ and $\theta$. 
Furthermore, there is the prospect that the simplified condition can also contribute to single-layered regularization (i.e., $\mathcal{R}^{(1,K)}(\mathbf{Ax}) = \mathcal{R}^{(1)}(\mathbf{Ax})$) as formulated in Appendix \ref{ap:singleReg}.

The ER-LiGME model achieves the non-convexification of the outermost function of $\mathcal{R}^{(1,K)}$, i.e., $\mathcal{R}^{(K)}$. As in Example 1, the Moreau envelope of $\mathcal{R}^{(K)}$ is preserved through the extension of $\mathbf{\Phi}$ with GOE. Consequently, our model successfully enhances the sparsity/low-rankness evaluated in $\mathcal{R}^{(K)}$.
\begin{Ex}
	In the ER-LiGME model \eqref{eq:prob1cELG}, the 2-layered regularization problem in \eqref{eq:problemTV} is reformulated as
	\begin{align}\label{eq:prob2cELG}
		\argmin_{\widehat{\mathbf{x}}} \frac{1}{2}&\left\|\begin{bmatrix}
			\boldsymbol{\Phi} & \mathbf{O} \\
			\mathbf{O} & \sqrt{\rho}\mathbf{I}
		\end{bmatrix}\widehat{\mathbf{x}} - \begin{bmatrix}
			{\mathbf{y}}  \\
			\sqrt{\rho}\ \mathbf{z}_{\mathrm{R}}
		\end{bmatrix} \right\|_2^2 \nonumber\\ &\hspace{5em}+ \mu \widehat{\mathcal{R}}_{\mathbf{B}}^{(2)}(\widehat{\mathbf{x}}) + \iota_{\mathrm{epi}_{\|\cdot\|_2}}(\mathbf{C}\widehat{\mathbf{x}}),
	\end{align}
	where $\B = \mathrm{diag}(\mathbf{O},\,\sqrt{\theta/\mu}\mathbf{I})$ defined as in \eqref{eq:settingB}. Then, the enhanced regularizer $\widehat{\mathcal{R}}_{\mathbf{B}}^{(2)}$ is derived as follows:
	\begin{align}\label{eq:PreservingME_inTV}
		\widehat{\mathcal{R}}_{\mathbf{B}}^{(2)}(\widehat{\mathbf{x} })& := 
		\|\mathbf{z}\|_1 - \min_{\widehat{\mathbf{v}}_2 } \left\{ \|\mathbf{v}_2\|_1 + \frac{\theta}{2\mu}\|{\mathbf{z}} - {\mathbf{v}_2}\|_2^2\right\}\nonumber\\ &= (\|\cdot\|_1)_{\sqrt{\theta/\mu}\mathbf{I}}(\mathbf{z}).
	\end{align}
	Unlike \eqref{eq:problemTV3}, the Moreau envelope is preserved in \eqref{eq:PreservingME_inTV}.
\end{Ex}

A global optimal solution of Problem \eqref{eq:prob1cELG} can be obtained using Algorithm \ref{alg:cLiGME}, as long as the closed-form proximity operators for ${\mathcal{R}}^{(K)}$, $\iota_{\mathrm{epi}_{\mathcal{R}^{(k)}}} (1\leq\forall k\leq K-1)$, and $\iota_{\mathcal{C}}$ are available. 

In the rest of this section, we discuss what multilayered non-convex regularization problem corresponds to the epigraphically-relaxed problem in \eqref{eq:prob1cELG} as the original one and analyze conditions ensuring the equivalence of their global optimal solution sets for deeper comprehension. First, we consider Problem \eqref{eq:genform} and its ER
problem in \eqref{eq:reform4} with an additional data term, and its equivalence
as in the following proposition.
\begin{prop}[Equivalence between \eqref{eq:genform} and \eqref{eq:reform4} with an additional data term]\label{Theorem2}
	Let $\mathcal{R}^{(1,K)}$ be a non-convex multilayered regularizer with strictly-increasing (blockwise-vector) function $\mathcal{R}^{(k)}$ ($2\leq k \leq K$), and consider the following problems: 
	\begin{align}\label{eq:T2_LiGME}
		&\mathcal{S}_\mathbf{x} :=\ 
		\argmin_{\mathbf{x}\in \mathcal{C} \subset \mathbb{R}^{N}} \mathcal{U}(\x) + \frac{\rho}{2} \| \mathcal{R}^{(1,K-1)}(\mathbf{A}\mathbf{x}) - {\mathbf{z}_{\mathrm{R}}}\|_2^2,\\
		&\label{eq:T2_LiGME2}\widetilde{\mathcal{S}}_{\mathbf{x}}\hspace{-0.25em}\times\hspace{-0.25em}\prod_{k=K}^{2}\widetilde{\mathcal{S}}_{\mathbf{z}^{(k)}}  := \argmin_{( \mathbf{x}, \{\mathbf{z}^{(k)}\}_{k=K}^{2}) \in \widetilde{\mathcal{D}} } \widetilde{\mathcal{V}}\left( \mathbf{x}, \{\mathbf{z}^{(k)}\}_{k=K}^{2}\right)\nonumber\\ 
		& \hspace{14em} + \frac{\rho}{2} \| \z^{(K)} - {\mathbf{z}_{\mathrm{R}}}\|_2^2,
	\end{align}
	where the functions $\mathcal{U}$ and $\widetilde{\mathcal{V}}$ are defined in \eqref{eq:genform} and \eqref{eq:reform4}, respectively.
	Then, for $\forall( \mathbf{x}^\star, \{\mathbf{z}^{\star(k)}\}_{k=K}^{2})\in\widetilde{\mathcal{S}}_{\mathbf{x}}\times\prod_{k=K}^{2}\widetilde{\mathcal{S}}_{\mathbf{z}^{(k)}}$,
	\begin{enumerate}
		\item if $\mathbf{z}_{\mathrm{R}} = \mathbf{0}$, then $( \mathbf{x}^\star, \{\mathbf{z}^{\star(k)}\}_{k=K}^{2})\in\mathcal{D}$ and $\mathcal{S}_{\mathbf{x}} = \widetilde{\mathcal{S}}_{\mathbf{x}}$,
		\item if \(\mathbf{z}_{\mathrm{R}} \neq \mathbf{0}\) and \([\mathbf{z}_{\mathrm{R}}]_i < [{\mathbf{z}}^{\star(K)}]_i\) holds for any \(i\) where \([\mathbf{z}_{\mathrm{R}}]_i \neq 0\), then $( \mathbf{x}^\star, \{\mathbf{z}^{\star(k)}\}_{k=K}^{2})\in\mathcal{D}$ and $\mathcal{S}_{\mathbf{x}} = \widetilde{\mathcal{S}}_{\mathbf{x}}$.
	\end{enumerate}
	where $\mathcal{D}$ is the set defined in \eqref{eq:reform2}.
\end{prop}
\begin{proof}
	see Appendix \ref{ap:poof_ERLiGME}.
\end{proof}
The optimization problem in \eqref{eq:prob1cELG} and its original (non-ER)
one correspond to the ones in \eqref{eq:T2_LiGME2} and \eqref{eq:T2_LiGME} with $\mathcal{F}_\y(\boldsymbol{\Phi}\mathbf{x}) = \tfrac{1}{2}\|\boldsymbol{\Phi}\mathbf{x}-\mathbf{y}\|_2^2$ and $\mathcal{R}_\B^{(K)}$ as the $K$-th layer function, respectively. Since the LiGME penalty is non-decreasing in general (e.g., the scaled MC penalty in \cite{SelesnickGMC2017}), the equivalence of their global optimal solution sets does not hold even for the above cases in Proposition \ref{Theorem2}. However, we can guarantee the equivalence by using a slightly modified LiGME penalty with the strictly increasing property, instead of using the LiGME penalty itself. For example, adding a strictly increasing function to the LiGME penalty at the $K$-th layer in \eqref{eq:prob1cELG}, i.e., $\mathcal{R}^{(K)}_{\mathbf{B},\varepsilon}:=\widehat{\mathcal{R}}^{(K)}_{\mathbf{B}}+\varepsilon\|\cdot\|_1\,(\varepsilon\in\R_{++},\,\varepsilon\sim 0)$, makes it strictly increasing.

\subsection{Setting of the Reference Data for the ER-LiGME Model}\label{sec:referenceData}
Concerning Remark \ref{rem:dataFidelty}, setting a reasonable reference data is important to estimate $\z^{(K)}$ accurately. One way is to set the reference data as  $\mathbf{z}_{\mathrm{R}}=\mathcal{R}^{(1,K-1)}(\mathbf{A}\mathbf{x}^\star)$, where $\x^\star$ is a solution for the non-LiGME problem in \eqref{eq:P1ER_reform}. However, it should be noted that $\mathcal{R}^{(1,K-1)}(\mathbf{A}\mathbf{x}^\star)$ suffers from underestimation when $\mathcal{R}^{(K)}$ is a convex sparsity/low-rankness promoting regularizer, such as the $\ell_1$-/nuclear-norm. Therefore, we set $\mathbf{z}_{\mathrm{R}} = \alpha \mathcal{R}^{(1,K-1)}(\mathbf{A}\mathbf{x}^\star)$, where $\alpha = 1/\rho$ $(0<\rho<1)$ or $\alpha = 1 + 1/\rho$ $ (1 \leq \rho)$.

\subsection{Computational Complexity of the ER-LiGME Model}\label{sec:complexity}
Compared with the cLiGME model, the main computational bottleneck of the ER-LiGME model with Algorithm \ref{alg:cLiGME} is the proximity operators generated by ER.
Nevertheless, as long as each proximal operator is closed-form, we can efficiently compute them for each step. The computational complexities of commonly used regularizers are presented in Table \ref{tab:complexity}. For the regularizers listed in the table, the computational complexity is within the range of $\mathcal{O}(N)$ to $\mathcal{O}(MN \min\{M,N\})$.

Solving \eqref{eq:P1ER_reform} for the reference data $\mathbf{z}_{\mathrm{R}}$ also increases the computational cost of the ER-LiGME model. One promising approach is to use a solution from a simplified optimization problem, e.g., $
\mathbf{x}^\star= \argmin_{\mathbf{x} \in \mathbb{R}^{N^2}} \frac{1}{2}\| \boldsymbol{\Phi}\mathbf{x} - \mathbf{y} \|_2^2 + \frac{\mu}{2} \|\mathbf{D}_{\mathrm{vh}}\mathbf{x}\|_{2}^2= (\boldsymbol{\Phi}^\top\boldsymbol{\Phi}+\mu\mathbf{D}_{\mathrm{vh}}^\top\mathbf{D}_{\mathrm{vh}})^{-1}\boldsymbol{\Phi}^\top\mathbf{y}
$. 
Although we adopt the solution of \eqref{eq:P1ER_reform} for the experiments in Sec. \ref{sec:experim}, we establish the simplified pre-computation for $\mathbf{z}_{\mathrm{R}}$ as our future work.

\subsection{Practical Applications of the ER-LiGME Model for Non-Proximable Multilayered Regularizers}\label{sec:ERLiGME_examples}
\subsubsection{ER-LiGME Decorrelated STV}\label{subsec:ERLiGME_DSTV}
The decorrelated STV (DSTV) is a 3-layered mixed norm modeled local gradient similarity by the nuclear norm of a local gradient matrix in luma and chroma components \cite{Kyochi2021}. It shows better performance in image recovery compared with other TV and STV based methods\cite{Bredies2010,Ono2014DVTV,Lefkimmiatis2015,Ono2016}. Its definition is as follows (see Appendix \ref{ap:examples_MR} for detailed information):
\begin{align}\label{eq:defDSTV-IV}
	\|\mathbf{x}\|_{\mathrm{DSTV}}=&\| \cdot \|_{2, 1}^{(1,2)} \circ \|\cdot\|_\ast^{(W^2)}(\mathbf{PED}\mathbf{C}_{\mathrm{D}}\mathbf{x}).
\end{align}
In order to enhance the DSTV by the ER-LiGME model, we set $\widehat{\mathbf{x}} = \begin{bmatrix}
	\mathbf{x}^\top & \mathbf{z}^{(2)} 
\end{bmatrix}^\top$, $\widehat{\boldsymbol{\Phi}}=\mathrm{diag}(\boldsymbol{\Phi},\,\sqrt{\rho}\mathbf{I})$, $\widehat{\mathbf{y}} = \begin{bmatrix}
\mathbf{y}^\top & \sqrt{\rho}\ \mathbf{z}_{\mathrm{R}}^{\top}
\end{bmatrix}^\top$, $ \widehat{\mathcal{R}}_{\mathbf{B}}^{(2)}(\widehat{\mathbf{x}}) = (\| \cdot \|_{2, 1}^{(1,2)})_\mathbf{B}(\mathbf{x})$, and $\iota_{\mathrm{epi}_{\mathcal{R}^{(1)}}} = \iota_{\mathrm{epi}_{\|\cdot\|_\ast^{(W^2)}}}$ in \eqref{eq:prob1cELG} as
\begin{align}\label{eq:prob4ERLiGME}
	&\argmin_{\widehat{\mathbf{x}}\in\R^{3N^2}} \frac{1}{2}\|\,\widehat{\boldsymbol{\Phi}}\widehat{\mathbf{x}} - \widehat{\mathbf{y}}\|_2^2 + \mu \widehat{\mathcal{R}}_{\mathbf{B}}^{(2)}(\widehat{\mathbf{x}}) + \iota_{\mathrm{epi}_{\mathcal{R}^{(1)}}\cap\,\mathcal{C}}(\mathbf{C}\widehat{\mathbf{x}}).
\end{align}
Even though the closed-form proximity operator of \eqref{eq:defDSTV-IV} is not obvious, the enhanced DSTV minimization can be tractable as the proximity operator for $\| \cdot \|_{2,1}^{(1,2)}$ and the projection onto epigraph of $\|\cdot\|_\ast^{(W^2)}$ can be computed in closed-form.

\subsubsection{ER-LiGME Amplitude Spectrum Nuclear Norm}
The robust PCA (RPCA) \cite{Candes2011} aims to extract principal components represented as low-rank components $\mathbf{L}\in\R^{M\times N}$ from observations $\mathbf{Y} = \mathbf{L}+\mathbf{S}\in\R^{M\times N}$ in the presence of outliers $\mathbf{S}\in\R^{M\times N}$, which can be
modeled by the following least-squares problem:
\begin{align}\label{eq:prob3}
	&\argmin_{\mathbf{x} \in \mathbb{R}^{2MN},\boldsymbol{\ell}\in\mathcal{C}} \frac{1}{2}\left\|\begin{bmatrix}
		\mathbf{I} & \mathbf{I}
	\end{bmatrix}\mathbf{x} - \mathbf{y} \right\|_2^2 + \mu f(\mathbf{x}), \nonumber\\
	&\mathbf{x} = \begin{bmatrix}
		\boldsymbol{\ell}^\top & \mathbf{s}^\top
	\end{bmatrix}^\top,\  f(\mathbf{x}) = \lambda_1\| \mathrm{vec}^{-1}(\boldsymbol{\ell}) \|_* + \lambda_2\varphi(\mathbf{s}),
\end{align}
where $\boldsymbol{\ell} = \mathrm{vec}(\mathbf{L}) $, $\mathbf{y} = \mathrm{vec}(\mathbf{Y} )$, and  $\mathbf{s} = \mathrm{vec}(\mathbf{S}) \in \R^{MN}$ are the vectorized data of the principal components, the observation, and the outliers, $\mathcal{C}\subset\R^{MN}$, $\lambda_1,\lambda_2\in\R_{++}$, and $\varphi(\cdot)$ is a sparsity-promoting regularizer. Here, we set $\lambda_2 = 1$ and $\varphi(\mathbf{s}) := \iota_{\|\cdot\|_1\leq \varepsilon_\mathbf{s}}(\mathbf{s})$ ($\varepsilon_\mathbf{s} \in \R_{++}$), and formulate the cLiGME problem corresponding to \eqref{eq:prob3} with $\boldsymbol{\Phi}=\begin{bmatrix}
	\mathbf{I} & \mathbf{I}
\end{bmatrix}$ as follows:
\begin{align}\label{eq:prob3cL}
\argmin_{\mathbf{x}\in\R^{2MN}}&\frac{1}{2}\left\|\boldsymbol{\Phi}\mathbf{x} - \mathbf{y} \right\|_2^2 + \mu \mathcal{R}_{\mathbf{B}}(\mathbf{x}) + \iota_{\mathcal{C}_2}(\mathbf{x}), \nonumber\\
&\mathcal{R}(\mathbf{x}) = \lambda_1\| \mathrm{vec}^{-1}_{(M,N)}(\boldsymbol{\ell}) \|_*,\nonumber\\
&\iota_{\mathcal{C}_2}(\mathbf{x})=\iota_{\mathcal{C}}(\boldsymbol{\ell}) + \iota_{\|\cdot\|_1\leq\varepsilon_\mathbf{s}}(\mathbf{s}),
\end{align}
where $\mathbf{B} = \mathrm{diag}(\sqrt{\theta/\mu}\I,\,\mathbf{O})\ (0\leq\theta\leq 1)$ that satisfies \eqref{eq:convcond}.

One of the drawbacks of RPCA is that it fails to extract principal components accurately in the presence of misalignment,  due to the use of the nuclear norm in the signal domain. The frequency domain RPCA (FRPCA) resolve this problem by using the amplitude spectrum nuclear norm (ASNN) \cite{Kyochi2021} that considers the fundamental singal processing property: the amplitude spectra of the shifted signals are identical. The ASNN can be regarded as a non-proximable 2-layered mixed norm involving a matrix $\mathbf{L}$ associated with the DFT (see Appendix \ref{ap:examples_MR} for detailed information):
\begin{align}\label{eq:defASNN_IV}
	\| \boldsymbol{\ell} \|_{\mathrm{ASNN}} = \|\mathrm{vec}^{-1}_{(M,N)}(\cdot)\|_* \circ \|\cdot\|_{2}^{(M,N)}(\mathbf{T}\boldsymbol{\ell}).
\end{align}

Let us consider the FRPCA problem with ASNN and its epigraphically-relaxed problem with $	\widehat{\mathbf{x}} = \begin{bmatrix}
	\boldsymbol{\ell}^\top & \mathbf{s}^\top & \mathbf{z}^\top
\end{bmatrix}^\top$ and $\boldsymbol{\Phi}_{\mathbf{O}}=\begin{bmatrix}
\mathbf{I} & \mathbf{I} & \mathbf{O}
\end{bmatrix}$ formulated as follows:
\begin{align}\label{eq:prob4}
	&\argmin_{\mathbf{x}\in\R^{2MN}}\frac{1}{2}\left\|\boldsymbol{\Phi}\mathbf{x} - \mathbf{y} \right\|_2^2 + \mu \mathcal{R}^{(1,2)}(\mathbf{x}) + \iota_{\mathcal{C}_2}(\mathbf{x}),\nonumber\\ 
	&\argmin_{\widehat{\mathbf{x}}\in\R^{3MN}} \frac{1}{2}\left\|\boldsymbol{\Phi}_{\mathbf{O}}\widehat{\mathbf{x}} - \mathbf{y} \right\|_2^2 + \mu \widehat{\mathcal{R}}^{(2)}(\widehat{\mathbf{x}})+ \iota_{\mathrm{epi}_{\|\cdot\|_2^{(M,N)}}\cap\,\mathcal{C}_2}(\mathbf{C}\widehat{\mathbf{x}}), \nonumber\\ 
	&\mathcal{R}^{(1,2)}(\mathbf{x}) = \lambda_1\| \boldsymbol{\ell} \|_{\mathrm{ASNN}},\ \widehat{\mathcal{R}}^{(2)}(\widehat{\mathbf{x}}) = \lambda_1\|\mathrm{vec}^{-1}_{(M,N)}(\mathbf{z})\|_{*},
\end{align}
where the matrix $\mathbf{C}$ maps $\widehat{\mathbf{x}}$ to $\mathbf{C}\widehat{\mathbf{x}} =\begin{bmatrix}
	\widehat{\mathbf{z}}^\top & \mathbf{x}^\top
\end{bmatrix}^\top$ with $\widehat{\mathbf{z}}= \begin{bmatrix} (\mathbf{T}\boldsymbol{\ell})^\top & \mathbf{z}^{\top}\end{bmatrix}^\top$. The ASNN can be enhanced by the ER-LiGME model as follows:
\begin{align}\label{eq:prob4cELG}
	\argmin_{\widehat{\mathbf{x}}\in\R^{3MN}} &\frac{1}{2}\left\|\widehat{\boldsymbol{\Phi}}\widehat{\mathbf{x}} - \widehat{\mathbf{y}} \right\|_2^2 + \mu \widehat{\mathcal{R}}_{\mathbf{B}}^{(2)}(\widehat{\mathbf{x}}) + \iota_{\mathrm{epi}_{\|\cdot\|_2^{(M,N)}}\cap\,\mathcal{C}_2}(\mathbf{C}\widehat{\mathbf{x}}),\nonumber\\
	&\widehat{\boldsymbol{\Phi}} =\begin{bmatrix}
		\mathbf{I} & \mathbf{I} & \mathbf{O}\\
		\mathbf{O} & \mathbf{O} & \sqrt{\rho}\mathbf{I}
	\end{bmatrix},\quad \widehat{\mathbf{y}} =\begin{bmatrix}
	{\mathbf{y}}  \\
	\sqrt{\rho}\ \mathbf{z}_{\mathrm{R}}
	\end{bmatrix}
\end{align}
where $\B = \mathrm{diag}(\mathbf{O},\,\mathbf{O},\,\sqrt{\theta/\mu}\mathbf{I})$. As is the case with ER-LiGME DSTV in Sec. \ref{subsec:ERLiGME_DSTV}, even though the closed-form proximity operator of \eqref{eq:defASNN_IV} is not obvious, the enhanced ASNN minimization can be tractable by using the closed-form proximity operator for $\|\cdot\|_*$, and the projections onto epigraph of $\|\cdot\|_2^{(M,N)}$ and the $\ell_1$-ball $\|\cdot\|_1\leq\varepsilon_\mathbf{s}$.

\section{Experiment Results}
\label{sec:experim}
We evaluated the performance of the ER-LiGME model in signal recovery through the three experiments: image denoising, compressed image sensing reconstruction, and principal component analysis from misaligned synthetic data in the presense of outlier. The accuracy was evaluated by the PSNR [dB] of estimated signals\footnote{PSNR is defined by $10\mathrm{log}_{10} (N/\|\mathbf{u}-\mathbf{u}_{\mathrm{org}}\|^2_2)$, where $\mathbf{u}\in\R^N$ is a restored signal, and $\mathbf{u}_{\mathrm{org}}\in[0,1]^N$ is the original one. The test data are normalized.}.
The images used in the experiments, shown in Fig. \ref{fig:ex1_images_org} and Fig. \ref{fig:ex2_images_org}, were provided from \textit{Berkeley Segmentation Database} (BSDS300) \cite{Martin2001}, and incomplete observations of them were generated as in the model of \eqref{eq:observation}, where $\n$ was defined as an additive white Gaussian noise which has a standard derivation $\sigma_{\mathrm{noise}}$.
The experiments were performed using MATLAB (R2023b, 64bit) on Windows OS (Version 11) with an Intel Core i9 3.7 GHz 10-core processor and 32 GB DIMM memory.

\begin{table}[t]
	\caption{Operators in Algorithm 1.}
	\begingroup
	\renewcommand{\arraystretch}{1.5}
	\setlength{\tabcolsep}{3pt}
	\centering
		\begin{tabular}{l|ll|l}
		\thline
		Regularizer   & $\boldsymbol{\mathfrak{L}}$                    & $\C$                  & $\B$      \\ \thline
		DVTV          & \multirow{2}{*}{$\D\mathbf{C}_{\mathrm{D}}$} & \multirow{2}{*}{$\I$} & $\mathbf{O}$     \\  \cline{1-1} \cline{4-4}
		LiGME DVTV    &   &  & $\begin{bmatrix} \sqrt{{\theta}/{\mu}}\widehat{\mathbf{M}} & \mathbf{O} \end{bmatrix}\mathbf{\Sigma}^\dagger\mathfrak{U}^\top$      \\ \hline
		ER-LiGME DVTV & $\I$                    & \renewcommand{\arraystretch}{1}$\begin{bmatrix}
			\D\mathbf{C}_{\mathrm{D}} & \mathbf{O} \\ \mathbf{O} & \I \\ \I & \mathbf{O}
		\end{bmatrix}$                & $\mathrm{diag}(\mathbf{O},\sqrt{\frac{\theta}{\mu}}\I)$ \\ \hline
		STV           & \multirow{2}{*}{$\P\E\D$} & \multirow{2}{*}{$\I$} & $\mathbf{O}$      \\ \cline{1-1} \cline{4-4}
		LiGME STV     &  &  & $\begin{bmatrix} \sqrt{{\theta}/{\mu}}\widehat{\mathbf{M}} & \mathbf{O} \end{bmatrix}\mathbf{\Sigma}^\dagger\mathfrak{U}^\top$      \\ \hline
		ER-LiGME STV  & $\I$                    & \renewcommand{\arraystretch}{1}$\begin{bmatrix}
			\P\E\D & \mathbf{O} \\ \mathbf{O} & \I \\ \I & \mathbf{O}
		\end{bmatrix}$               & $\mathrm{diag}(\mathbf{O},\sqrt{\frac{\theta}{\mu}}\I)$ \\ \hline
		DSTV          & \multirow{2}{*}{$\I$} & \multirow{2}{*}{\renewcommand{\arraystretch}{1}$\begin{bmatrix}
			\P\E\D\mathbf{C}_{\mathrm{D}} & \mathbf{O} \\ \mathbf{O} & \I \\ \I & \mathbf{O}
		\end{bmatrix}$} & $\mathbf{O}$      \\  \cline{1-1} \cline{4-4}
		ER-LiGME DSTV &  &  & $\mathrm{diag}(\mathbf{O},\sqrt{\frac{\theta}{\mu}}\I)$ \\ \cline{1-4}
		NN          & \multirow{2}{*}{$\I$}  & \multirow{2}{*}{$\I$} & $\mathbf{O}$      \\ \cline{1-1} \cline{4-4}
		LiGME NN    &  & &
		$\mathrm{diag}(\sqrt{\frac{\theta}{\mu}}\I, \mathbf{O})$ \\ \cline{1-4}
		ASNN          & \multirow{2}{*}{$\I$}  & \multirow{2}{*}{\renewcommand{\arraystretch}{1}$\begin{bmatrix}
			\mathbf{T} & \mathbf{O} & \mathbf{O} \\ \mathbf{O} & \mathbf{O} & \I \\ \I & \mathbf{O} & \mathbf{O}\\ \mathbf{O} & \I & \mathbf{O}
		\end{bmatrix}$} & $\mathbf{O}$      \\ \cline{1-1} \cline{4-4}
		ER-LiGME ASNN &  &  & $\mathrm{diag}(\mathbf{O},\mathbf{O},\sqrt{\frac{\theta}{\mu}}\I)$\rule{0pt}{2em} \\
		\thline
	\end{tabular}
	\label{tab:ex1_prob}
	\endgroup
	\vspace{0.5em}
\end{table}

\subsection{Denoising} \label{sec:exDenoising}
This section shows the results of denoising on the five images displayed in Fig. \ref{fig:ex1_images_org}. The size of each image was $63\times 63$ and the observation was given with $\boldsymbol{\Phi} = \mathbf{I}$ and $\sigma_{\mathrm{noise}}=0.1$. We compared with eight regularization methods based on the three regularizers: the decorrelated vectorial TV (DVTV)\cite{Ono2014DVTV}, STV\cite{Lefkimmiatis2015}, and DSTV. For the proximable regularizers, DVTV and STV, we conducted experiments in three cases: 1) their original form, 2) using in cLiGME model (termed L-DVTV/STV), and 3) using in ER-LiGME model (termed EL-DVTV/STV). For the non-proximable regularizer, DSTV, we experimented with two cases: 1) the original form and 2) using in ER-LiGME model (termed EL-DSTV). All of the optimization problems were solved by Algorithm \ref{alg:cLiGME}, where $\mathcal{C}:=[0,1]^{3N^2}$($N=63$). The forms of $\boldsymbol{\mathfrak{L}},\mathbf{C}$, and $\mathbf{B}$ for each method are summarized in Table \ref{tab:ex1_prob}, where $\mathbf{D},\mathbf{C}_{\mathrm{D}}, \P$, and $\E$ are detailed in Appendix \ref{ap:examples_MR}. The design of $\B$ for L-DVTV/STV were given by \cite[Proposition 2]{Chen2023}, consisting of the left singular matrix $\boldsymbol{\mathfrak{U}}\in\R^{N_{\boldsymbol{\mathfrak{L}}}\times N_{\boldsymbol{\mathfrak{L}}}}$ and the singular value matrix $\boldsymbol{\Sigma}\in \mathbb{R}^{N_{\boldsymbol{\mathfrak{L}}}\times N}$ for $\boldsymbol{\mathfrak{L}}\in \mathbb{R}^{N_{\boldsymbol{\mathfrak{L}}}\times N}$ ($\boldsymbol{\mathfrak{U}\Sigma\mathfrak{B}}^\top=\boldsymbol{\mathfrak{L}}$), and the matrix $\widehat{\mathbf{M}}=\boldsymbol{\Phi}_1-\boldsymbol{\Phi}_2 \boldsymbol{\Phi}_2^\dagger \boldsymbol{\Phi}_1\in\R^{M\times r}$, where $r=\mathrm{rank}(\boldsymbol{\mathfrak{L}})$, and the matrices $\boldsymbol{\Phi}_1\in\R^{M\times r}$ and $\boldsymbol{\Phi}_2\in\R^{M\times (N-r)}$ are given by $\begin{bmatrix}\boldsymbol{\Phi}_1\ |\  \boldsymbol{\Phi}_2\end{bmatrix}=\boldsymbol{\Phi}\boldsymbol{\mathfrak{B}}$.

\subsubsection{Parameter Settings}
The regularization parameter was set to  $\mu=0.05,\,0.1,\,0.2,\,0.3,\,0.4,\,0.5$. The Moreau enhancement parameters for the cLiGME model and ER-LiGME model were $\theta = 0.99$ and $\theta = \rho = 3.50$, respectively. In Algorithm \ref{alg:cLiGME}, $\tau$ and $\sigma$ were set to satisfy the condition \cite[Theorem 1]{Abe2020} with $\kappa=1.1$. The stop criteria was set to $\varepsilon_{\mathrm{stop}} = 10^{-3}$. The methods that employ STV or DSTV were implemented with non-overlapping $3\times 3$ patches\footnote{The computation of $\B$ proposed in \cite{Chen2023} for conventional LiGME model requires an explicit definition of $\boldsymbol{\mathfrak{L}}$ and $\boldsymbol{\Phi}$. However, the matrices may surpass the available memory capacity. For instance, when employing STV/DSTV with $W=3,N=63,$ and fully overlapping patches, $\boldsymbol{\mathfrak{L}}\in\R^{6W^2N^2\times 6W^2N^2}$ requires 171.1 GB of memory in single-precision floating-point format. To circumvent this, we implemented these methods without overlapping (then,  $\boldsymbol{\mathfrak{L}}\in\R^{6N^2\times 6N^2}$).}, i.e., $W=3$ and the expansion operator $\E=\I$.

\subsubsection{Results}
Table \ref{tab:ex1_PSNR} shows the maximum PSNR values for each method, chosen from the PSNRs corresponding to $\mu$ ranging from 0.05 to 0.5. The results indicate that EL-DSTV achieved the highest maximum PSNR throughout all of the test images.
Moreover, the EL-DSTV restored images with superior quality compared to those restored by the other methods. Fig. \ref{fig:ex1_images} shows that  DVTV-based methods (original/L-/EL-DVTV) suffered from the staircase effect, whereas STV-based methods (original/L-/EL-STV) exhibited color artifact. Compared with them, the EL-DSTV was free from these shortcomings. 

The effectiveness of the ER-LiGME model in reducing underestimation is depicted by the bar graphs shown in Fig. \ref{fig:ex1_histograms}, which are the histograms of $|\D\x^\star|$ in the pixels such that $|\D\x_{\mathrm{org}}|\geq 1$ at the pixel\footnote{For an $N\times N$ image $\x$, we evaluated $|\D\x|$ per pixel by $\|\mathbf{P}_{\mathrm{grad}}\D\x\|_1\in\R^{N^2}$, where $\mathbf{P}_{\mathrm{grad}}$ is a permutation matrix and $\|\cdot\|_1$ is a block-wise $\ell_1$-norm.}, where $\x^\star$ is the estimated image and $\x_{\mathrm{org}}$ is the original image. As the histograms in Fig. \ref{fig:ex1_histograms} shows, the ER-LiGME model effectively reduced underestimation for both proximable and non-proximable multilayered regularizers, whereas the LiGME model only for proximable ones.

\begin{figure}[t]
	\begin{minipage}{\linewidth}
		\centering
		\begin{subfigmatrix}{6}
			\subfigure[\footnotesize{Image 1}]{\includegraphics[width=0.156\linewidth]{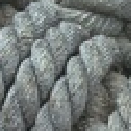}}
			\hspace{-0.75em}
			\subfigure[\footnotesize{\hspace{-0.45em}Observed}]{\includegraphics[width=0.156\linewidth]{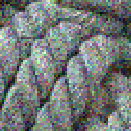}}
			\subfigure[\footnotesize{Image 2}]{\includegraphics[width=0.156\linewidth]{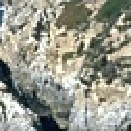}} 
			\hspace{-0.75em}
			\subfigure[\footnotesize{Image 3}]{\includegraphics[width=0.156\linewidth]{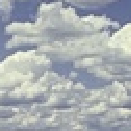}} 
			\hspace{-0.75em}
			\subfigure[\footnotesize{Image 4}]{\includegraphics[width=0.156\linewidth]{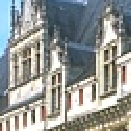}} 
			\hspace{-0.75em}
			\subfigure[\footnotesize{Image 5}]{\includegraphics[width=0.156\linewidth]{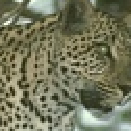}}
		\end{subfigmatrix}
		\caption{Test images for denoising.}
		\label{fig:ex1_images_org}
		\vspace{1em}
	\end{minipage}
	
	\begin{minipage}{\linewidth}
	\centering	
	\begin{subfigmatrix}{4}
		\subfigure[\parbox{0.1\linewidth}{DVTV\newline23.82}]{\includegraphics[width=0.248\linewidth]{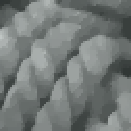}}
		\hspace{-1em}
		\subfigure[\parbox{0.15\linewidth}{L-DVTV\newline25.47}]{\includegraphics[width=0.248\linewidth]{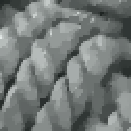}}
		\hspace{-1em}
		\subfigure[\parbox{0.15\linewidth}{EL-DVTV\newline25.37}]{\includegraphics[width=0.248\linewidth]{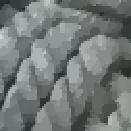}}
		\hspace{-1em}
		\subfigure[\parbox{0.1\linewidth}{STV\newline23.88}]{\includegraphics[width=0.248\linewidth]{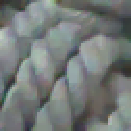}}
		\subfigure[\parbox{0.1\linewidth}{L-STV\newline24.43}]{\includegraphics[width=0.248\linewidth]{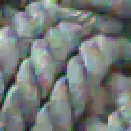}}
		\hspace{-1em}
		\subfigure[\parbox{0.15\linewidth}{EL-STV\newline24.55}]{\includegraphics[width=0.248\linewidth]{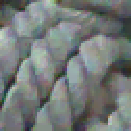}}
		\hspace{-1em}
		\subfigure[\parbox{0.1\linewidth}{DSTV\newline26.19}]{\includegraphics[width=0.248\linewidth]{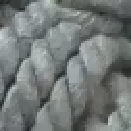}}
		\hspace{-1em}
		\subfigure[\parbox{0.15\linewidth}{\textbf{EL-DSTV}\newline \textbf{26.96}}]{\includegraphics[width=0.248\linewidth]{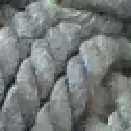}} 
	\end{subfigmatrix}
	\caption{Denoised images by each method ($\mu=0.3$, Image 1). PSNR [dB] is at the bottom of each image.}
	\label{fig:ex1_images}
	\end{minipage}
\end{figure}

\begin{table}[t]
	\caption{Maximum PSNRs [$\mathrm{dB}$] in denoising.}
	\centering
	\begingroup
	\setlength{\tabcolsep}{4.5pt}
	\begin{tabular}{l|rrrrr}
		\thline
		Method & Image 1 & Image 2 & Image 3 & Image 4 & Image 5 \\
		\thline
		DVTV & 26.73 & 25.20 & 28.06 & 25.36 & 25.77 \\
		L-DVTV & 26.20 & 25.00 & 27.72 & 25.52 & 25.61 \\
		EL-DVTV & 26.47 & 25.23 & 27.86 & 25.40 & 25.46 \\
		\hline
		STV & 24.81 & 22.52 & 26.79 & 22.74 & 23.16 \\
		L-STV & 24.43 & 22.49 & 26.48 & 23.29 & 23.02 \\
		EL-STV & 24.55 & 22.69 & 26.60 & 22.95 & 23.31 \\
		\hline
		DSTV & 26.93 & 25.28 & 28.37 & 25.62 & \textbf{25.93} \\
		\textbf{EL-DSTV} & \textbf{26.96} & \textbf{25.55} & \textbf{28.46} & \textbf{25.94} & \textbf{25.93} \\
		\thline
	\end{tabular}
	\endgroup
	\label{tab:ex1_PSNR}
\end{table}

\begin{figure}[tbhp]
\centering
\begin{subfigmatrix}{3}
	\subfigure[DVTV]{\includegraphics[width=0.3\linewidth]{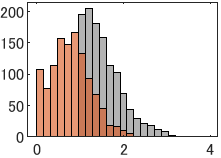}}
	\subfigure[STV]{\includegraphics[width=0.3\linewidth]{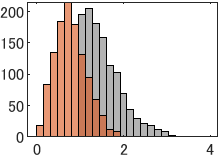}}
	\subfigure[DSTV]{\includegraphics[width=0.3\linewidth]{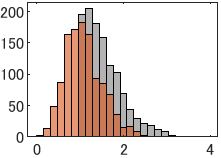}}
	\subfigure[L-DVTV]{\includegraphics[width=0.3\linewidth]{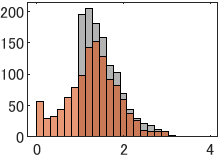}}
	\subfigure[L-STV]{\includegraphics[width=0.3\linewidth]{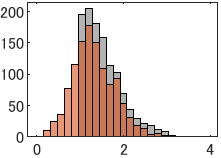}}
	\subfigure{\phantom{\includegraphics[width=0.3\linewidth]{figures/ex1/image4_1_Histgram_Dx__DSTV.eps}}} 
	\addtocounter{subfigure}{-1}
	\subfigure[EL-DVTV]{\includegraphics[width=0.3\linewidth]{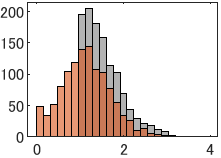}}
	\subfigure[EL-STV]{\includegraphics[width=0.3\linewidth]{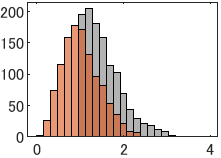}}
	\subfigure[\textbf{EL-DSTV}]{\includegraphics[width=0.3\linewidth]{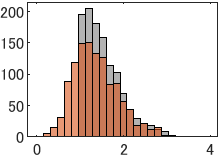}} 
\end{subfigmatrix}
\caption{Histograms of $|\D\x^\star|$ ($\mu=0.3$, Image 4). Gray graphs: ground truth, orange ones: estimated values by each method. }
\label{fig:ex1_histograms}
\end{figure}

\subsection{Compressed Sensing Reconstruction}

To provide another practical application of the ER-LiGME model, this section shows the results of compressed sensing reconstruction on the five images displayed in Fig. \ref{fig:ex2_images_org}. The size of each image was $128\times 128$ and the observation was given with $\boldsymbol{\Phi} = \mathbf{S}\widetilde{\mathbf{\Phi}}$ and $\sigma_{\mathrm{noise}}=0.1$, where $\widetilde{\mathbf{\Phi}}:\R^{3N^2}\rightarrow \R^{3N^2}$ is the noiselet transform \cite{COIFMAN2001} operator and $\mathbf{S} \in \mathbb{R}^{L\times 3N^2}$ $(L=0.4 \times 3N^2 )$ is a random downsampling matrix. We compared with six regularization methods based on the three regularizers: DVTV, STV, and DSTV. For all regularizers, we conducted experiments in two cases\footnote{Their LiGME versions were excluded from this experiment because they have difficulties in computing the matrix $\B$ with the large and dense observation process $\boldsymbol{\Phi}$ ($\neq\I$), which requires substantial memory for singular value decomposition to calculate $\boldsymbol{\Phi}_2^\dagger$. Even with the method using LDU decomposition\cite[Remark 1]{Chen2023}, the computational costs remain significant.}: 1) their original form and 2) using in ER-LiGME model (termed EL-DVTV/STV/DSTV). All of the optimization problems were solved by Algorithm \ref{alg:cLiGME}, where $\mathcal{C}:=[0,1]^{3N^2}$($N=128$), and the forms of  $\boldsymbol{\mathfrak{L}},\mathbf{C}$, and $\mathbf{B}$ for each method are listed in Table \ref{tab:ex1_prob}.

\subsubsection{Parameter Settings}
The regularization parameter $\mu$, the stop criteria $\varepsilon_{\mathrm{stop}}$, and the setting method of $\tau$ and $\sigma$ were the same as in Sec. \ref{sec:exDenoising}. The Moreau enhancement parameter for ER-LiGME model was set to $\theta = \rho = 4.00$. The methods that employ STV or DSTV were implemented with fully overlapping $3\times 3$ patches, i.e., $W=3,\,\E\in\R^{6W^2N^2\times 3N^2}$. The weights in the expansion matrix $\E$ were uniformly assigned as $1/W^2$.

\subsubsection{Results}
Table \ref{tab:ex2_PSNR} shows the maximum PSNRs with respect to the parameter $\mu$, similarly to Table \ref{tab:ex1_PSNR}.
The results demonstrate that EL-DSTV achieved the highest or comparable performance in all of the test images, even under the compression. Moreover, the images restored by EL-DSTV showed superior image quality to those restored by the other methods. Similar to Sec. \ref{sec:exDenoising}, the DVTV-based methods suffered from the staircase effect and the STV-based methods exhibited color artifact as shown in Fig. \ref{fig:ex2_images}. Compared with them, the EL-DSTV method overcame these problems.

The effectiveness of the ER-LiGME model in reducing underestimation is presented by the histograms shown in Fig. \ref{fig:ex2_histograms}, which are the histograms of $|\D\x^\star|$ in the pixels such that $|\D\x_{\mathrm{org}}|\geq 1$ at the pixel. As indicated in Fig. \ref{fig:ex2_histograms}, the ER-LiGME model effectively reduced underestimation for all regularizers. Notably, we can see that the histogram for EL-DSTV in Fig. \ref{fig:ex2_histograms} (f) is the closest to the true histogram.

\begin{figure}[t]
	\begin{minipage}{\linewidth}
		\centering
		\begin{subfigmatrix}{6}
			\subfigure[\footnotesize{Image 6}]{\includegraphics[width=0.156\linewidth]{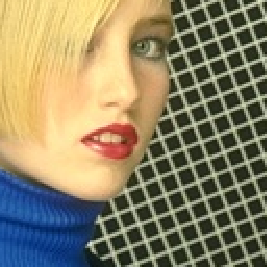}}
			\hspace{-0.75em}
			\subfigure[\footnotesize{\hspace{-0.45em}Observed}]{\includegraphics[width=0.156\linewidth]{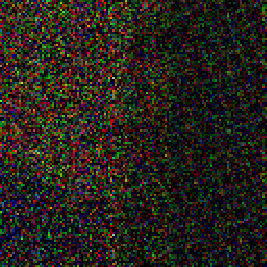}}
			\hspace{-0.5em}
			\subfigure[\footnotesize{Image 7}]{\includegraphics[width=0.156\linewidth]{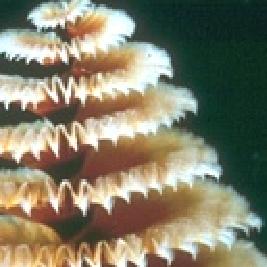}} 
			\hspace{-0.75em}
			\subfigure[\footnotesize{Image 8}]{\includegraphics[width=0.156\linewidth]{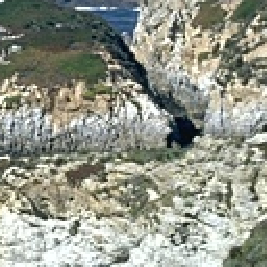}} 
			\hspace{-0.75em}
			\subfigure[\footnotesize{Image 9}]{\includegraphics[width=0.156\linewidth]{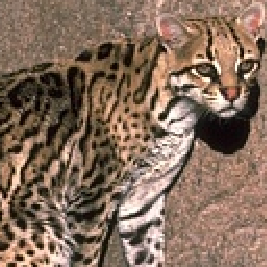}} 
			\hspace{-0.75em}
			\subfigure[\parbox{0.7\linewidth}{\footnotesize{\hspace{-0.25em}Image\hspace{0.125em}10}}]{\includegraphics[width=0.156\linewidth]{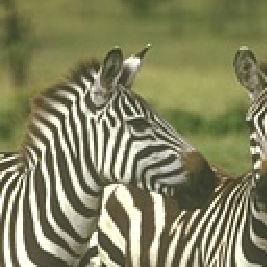}}
		\end{subfigmatrix}
		\caption{Test images for compressed sensing reconstruction.}
		\label{fig:ex2_images_org}
		\vspace{1em}
	\end{minipage}
	
	\begin{minipage}{\linewidth}
	\centering
	\begin{subfigmatrix}{3}
		\subfigure[\parbox{0.1\linewidth}{DVTV\newline23.64}]{\includegraphics[width=0.33\linewidth]{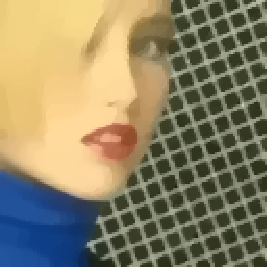}}
		\hspace{-0.5em}
		\subfigure[\parbox{0.1\linewidth}{STV\newline23.17}]{\includegraphics[width=0.33\linewidth]{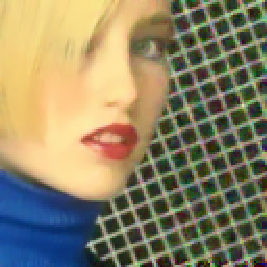}}
		\hspace{-0.5em}
		\subfigure[\parbox{0.1\linewidth}{DSTV\newline30.54}]{\includegraphics[width=0.33\linewidth]{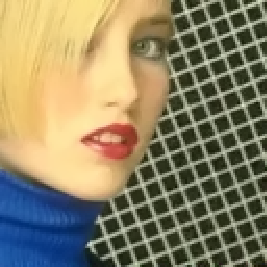}}
		\subfigure[\parbox{0.15\linewidth}{EL-DVTV\newline26.33}]{\includegraphics[width=0.33\linewidth]{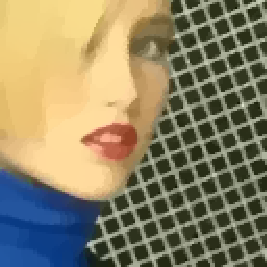}}
		\hspace{-0.5em}
		\subfigure[\parbox{0.15\linewidth}{EL-STV\newline23.96}]{\includegraphics[width=0.33\linewidth]{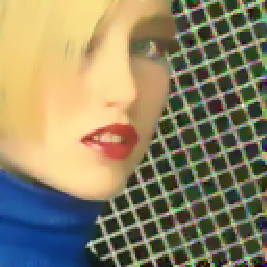}}
		\hspace{-0.5em}
		\subfigure[\parbox{0.15\linewidth}{\textbf{EL-DSTV}\newline\textbf{32.63}}]{\includegraphics[width=0.33\linewidth]{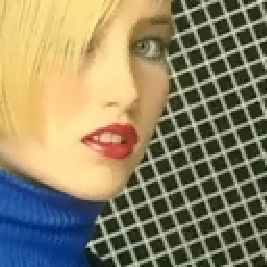}}
	\end{subfigmatrix}
	\caption{Reconstructed images by each method ($\mu=0.1$, image 6). PSNR [dB] are shown at the bottom of each image.}
	\label{fig:ex2_images}
	\end{minipage}
\end{figure}

\begin{table}[t]
\caption{Maximum PSNRs [$\mathrm{dB}$] in compressed sensing reconstruction.}
\centering
\begingroup
\setlength{\tabcolsep}{4.5pt}
\begin{tabular}{l|rrrrr}
	\thline
	Method & Image 6 & Image 7 & Image 8 & Image 9 & Image 10 \\
	\thline
	DVTV & 26.67 & 28.26 & 23.49 & 24.44 & 27.23 \\
	EL-DVTV & 29.17 & 30.39 & 25.04 & 26.07 & 30.53 \\
	\hline
	STV & 23.53 & 29.23 & 20.89 & 22.73 & 25.48 \\
	EL-STV & 23.96 & 29.45 & 20.79 & 22.33 & 25.71 \\
	\hline
	DSTV & 31.57 & 34.02 & 26.95 & \textbf{28.52} & 34.36 \\
	\textbf{EL-DSTV} & \textbf{32.64} & \textbf{34.11} & \textbf{27.44} & 28.49 & \textbf{35.56} \\
	\thline
\end{tabular}
\endgroup
\label{tab:ex2_PSNR}
\end{table}

\begin{figure}[tbhp]
	\centering
	\begin{subfigmatrix}{3}
		\subfigure[DVTV]{\includegraphics[width=0.3\linewidth]{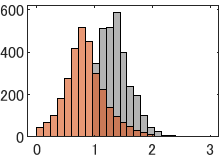}}
		\subfigure[STV]{\includegraphics[width=0.3\linewidth]{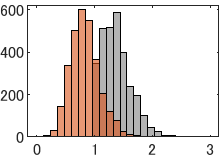}}
		\subfigure[DSTV]{\includegraphics[width=0.3\linewidth]{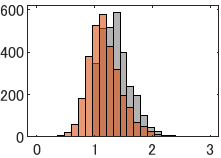}}
		\subfigure[EL-DVTV]{\includegraphics[width=0.3\linewidth]{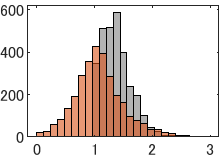}}
		\subfigure[EL-STV]{\includegraphics[width=0.3\linewidth]{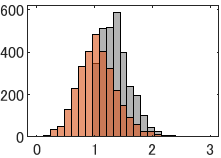}}
		\subfigure[\textbf{EL-DSTV}]{\includegraphics[width=0.3\linewidth]{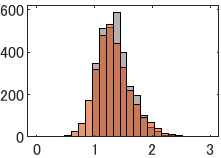}}
	\end{subfigmatrix}
	\caption{Histograms of $|\mathbf{D}\mathbf{x}^\star|$ ($\mu=0.1$, Image 6). Gray graphs: ground truth, orange ones: estimated values by each method. }
	\label{fig:ex2_histograms}
\end{figure}

\subsection{Principal Component Analysis}
As a further instance of practical applications, we demonstrated the extraction of principal components from the noisy signal $\mathbf{Y}=\mathbf{L}_s+\mathbf{S}\in \R^{M\times N}$, consisting of $N$ $M$-dimensional vectors obtained from $N$-times observations as shown in Fig. \ref{fig:ex3_images} (b), where $n$-th column of $\mathbf{L}_s$ is set as
\begin{align}
	[\mathbf{L}_s]_{m,n} = \begin{cases}
		1 & (sn\leq m \leq sn+7) \\
		0 & (\mathrm{otherwise})
	\end{cases}.
\end{align}
The test dataset includes 1) unshifted signals, 2) signals shifted by 1 sample (1-shift), and 3) signals shifted by 2 samples (2-shift), denoted as  $\mathbf{L}_0$,  $\mathbf{L}_1 \in \R^{32\times25}$, and  $\mathbf{L}_2\in \R^{32\times13}$, respectively. The observation was given with the sparse outliers $\mathbf{S}\in\R^{M\times N}$ that contains $1$s in the region of $\mathcal{N}_s = \{(m,n) \in \{1,2,\ldots,M\}\times\{1,2,\ldots,N\} \ |\  [\mathbf{L}_s]_{m,n}=0\}$ with the occurrence probability \(p_{\mathrm{noise}} = 0.1\). We evaluated four different regularization methods: the nuclear-norm (NN), the ASNN, the cLiGME NN (termed L-NN), and the ER-LiGME ASNN (termed EL-ASNN). All optimization problems were solved using Algorithm \ref{alg:cLiGME}, with $\mathcal{C}:=[0,1]^{MN}$. The forms of \(\boldsymbol{\mathfrak{L}}, \mathbf{C}\), and $\mathbf{B}$ for each method are listed in Table \ref{tab:ex1_prob}, where $\mathbf{T}$ is detailed in Appendix \ref{ap:examples_MR}.

\subsubsection{Parameter Settings}
The regularization parameter was set to $\mu=1$,  $\lambda_1=0.05,\,0.25,\,0.50,\,0.75,\,1.00,\,1.25$, and $\varepsilon_\mathbf{s}=MNp_{\mathrm{noise}}$. The Moreau enhancement parameter for the cLiGME model was set experimentally to $\theta = 0.10$, taking into account the balance between noise reduction and underestimation reduction, and that for the ER-LiGME model was $\theta = \rho = 1.00$. The stop criteria was set to $\varepsilon_{\mathrm{stop}} = 10^{-5}$. The setting method of $\tau$ and $\sigma$ were the same as in Sec. \ref{sec:exDenoising}.

\subsubsection{Results}
Despite the significant decrease of PSNR in the NN-based methods for shifted signals, the ASNN-based methods robustly maintained the performance against the shifts as shown in Table \ref{tab:ex3_PSNR}, which shows the maximum PSNRs selected from the results obtained for each the parameter $\lambda_1$. Then, the EL-ASNN achieved the highest maximum PSNR in both of the 1-shifted and 2-shifted signals. Fig. \ref{fig:ex3_images} shows that the EL-ASNN successfully removed the outliers.

The effectiveness of the ER-LiGME model in reducing underestimation is presented by the line graphs shown in Fig. \ref{fig:ex3_Linegraphs}, which are the sliced views of the extracted components at the half columns. As illustrated in Fig. \ref{fig:ex3_Linegraphs}, the EL-ASNN reduced underestimation in the ASNN, and the shape of the principal components was the closest to the true one.

\begin{figure}[t]
	\centering
	\begin{subfigmatrix}{3}
		\subfigure[Original]{\includegraphics[width=0.3\linewidth]{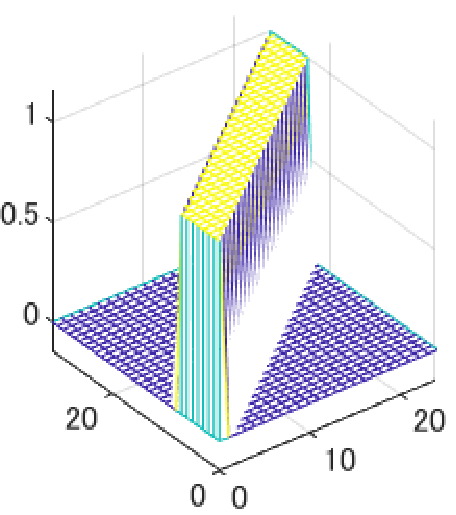}}
		\subfigure[Observed]{\includegraphics[width=0.3\linewidth]{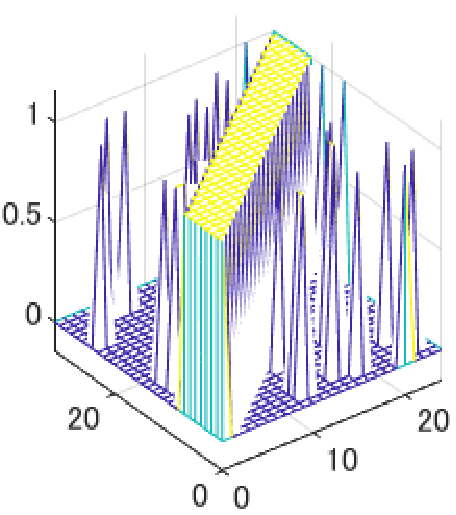}}
		\subfigure[\parbox{0.1\linewidth}{NN\newline13.84}]{\includegraphics[width=0.3\linewidth]{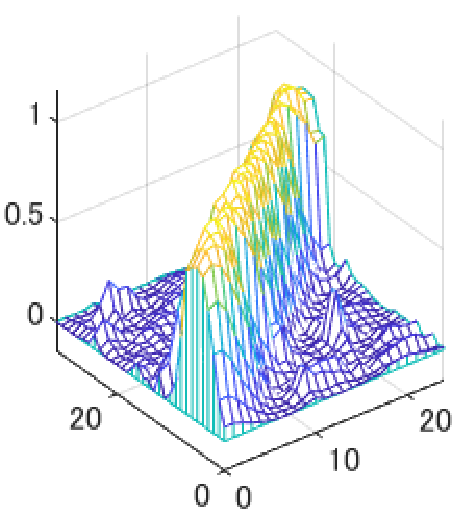}} 
		\subfigure[\parbox{0.1\linewidth}{L-NN\newline14.25}]{\includegraphics[width=0.3\linewidth]{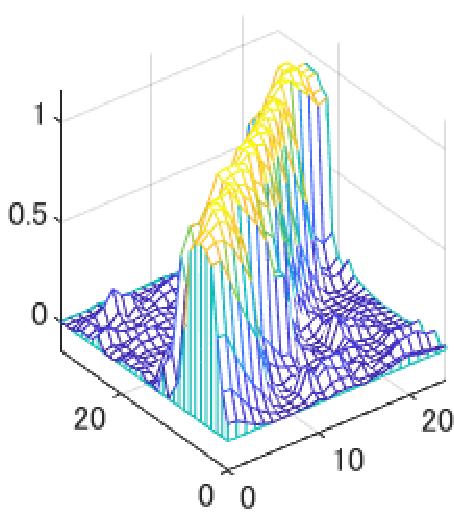}}
		\subfigure[\parbox{0.1\linewidth}{ASNN\newline20.17}]{\includegraphics[width=0.3\linewidth]{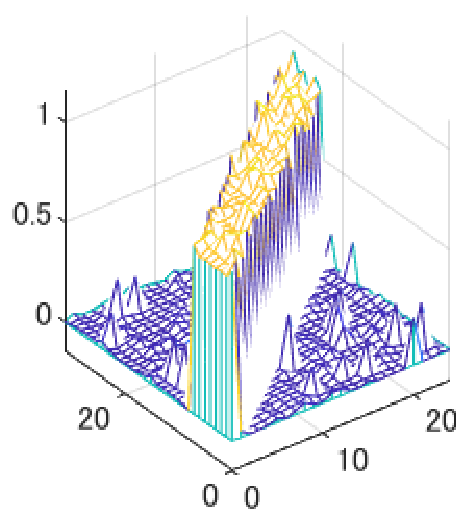}}
		\subfigure[\parbox{0.15\linewidth}{\textbf{EL-ASNN}\newline\textbf{22.17}}]{\includegraphics[width=0.3\linewidth]{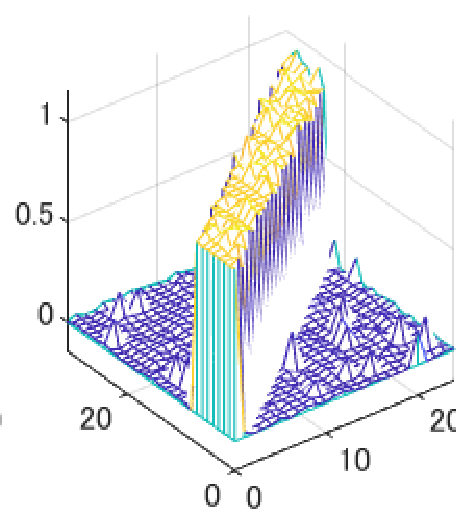}}
	\end{subfigmatrix}
	\caption{Extracted components of 1-shifted signal and PSNR [dB] ($\lambda_1=1.00$).}
	\label{fig:ex3_images}
	
	\vspace{1.5em}
	\includegraphics[width=0.75\linewidth,keepaspectratio=true]{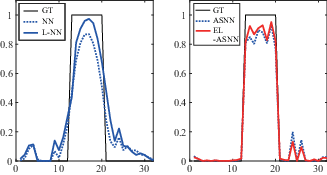}
	\caption{Sliced view of the extracted components of 1-shifted signal at the half column (GT: Ground truth, $\lambda_1=1.00$).}
	\label{fig:ex3_Linegraphs}
\end{figure}

\begin{table}[t]
\centering
\caption{Maximum PSNRs [$\mathrm{dB}$] in RPCA.}
\begingroup
\setlength{\tabcolsep}{4.5pt}
\begin{tabular}{l|llll}
	\thline
	& NN             & L-NN       & ASNN  & \textbf{EL-ASNN}  \\ \thline
	Shift: 0 & 34.42 & \textbf{35.86} & 24.59 & 24.66          \\
	Shift: 1 & 16.30          & 14.60          & 24.33 & \textbf{24.40} \\
	Shift: 2 & 15.53          & 14.93          & 23.43 & \textbf{23.51}
	\\\thline
\end{tabular}
\endgroup
\label{tab:ex3_PSNR}
\end{table}

\section{Conclusion}
\label{sec:conc}
In this paper, we proposed the signal recovery method with epigraphically-relaxed multilayered non-convex regularization, and the ER-LiGME model with GOE as a CPNR. In addition, we analyzed in Proposition \ref{Theorem1} and \ref{Theorem2} that applying ER to a multilayered non-convex regularization problem results in the same solution set as the original problem under the strictly-increasing property. Finally, experimental results showed that the two new non-convex regularizers, the ER-LiGME DSTV and ER-LiGME ASNN functions, not only reduced bias but also achieved high performance in denoising, compressed sensing reconstruction, and principal component analysis.


%

\appendices

\section{Proof of Lemma \ref{Lemma1}}\label{ap:poof_lemma1}
Let us prove Lemma \ref{Lemma1} by contradiction. Suppose $\mathcal{R}^{(k)}({\mathbf{z}}^{\star(k)}) \neq {\mathbf{z}}^{\star(k+1)}$ for $k=K-1$.
Then, there exists an index set $\mathcal{I}_K\subset\{1,...,N_K\}$ satisfying $\mathcal{R}^{(K-1)}({\mathbf{z}}^{\star(K-1)})\lneqq_{\mathcal{I}_K}{{\z}}^{\star(K)}$. For $\forall i_K\in\mathcal{I}_K$, we can choose $\varepsilon_{i_K}  \in (0,1)$ satisfying 
\begin{align} [\mathcal{R}^{(K-1)}({\mathbf{z}}^{\star(K-1)})]_{i_K} <  \varepsilon_{i_K} {z}_{i_K}^{\star(K)} < {z}_{i_K}^{\star(K)}.
\end{align}
Define $\widecheck{\mathbf{z}}^{(K)}$ by replacing the $i_K$-th element of ${\mathbf{z}}^{\star(K)}$ as
\begin{align}\label{eq:ztild1}
	\widecheck{\mathbf{z}}^{(K)} := \begin{bmatrix}  \ldots & {z}_{i_K-1}^{\star(K)} & \varepsilon_{i_K} {z}_{i_K}^{\star(K)} & {z}_{{i_K}+1}^{\star(K)} & \ldots \end{bmatrix}^\top.
\end{align} 
Since $\widecheck{\mathbf{z}}^{(K)}\lneqq_{\{i_K\}}{\mathbf{z}}^{(K)}$ and $\mathcal{R}^{(K)}$ has the strictly increasing property,  $\mathcal{R}^{(K)}(\widecheck{\mathbf{z}}^{(K)})<\mathcal{R}^{(K)}({{\mathbf{z}}}^{\star(K)})$ holds. This contradicts that ${{\mathbf{z}}}^{\star(K)}$ is the minimizer of Problem \eqref{eq:reform4}. 

Next, suppose that $\mathcal{R}^{(k)}({\mathbf{z}}^{\star(k)}) = {\mathbf{z}}^{\star(k+1)}$ holds for $k = K-1$, and does not for $k=K-2$, i.e., $\mathcal{R}^{(K-1)}({\mathbf{z}}^{\star(K-1)}) =  {\mathbf{z}}^{\star(K)},\ \mathcal{R}^{(K-2)}({\mathbf{z}}^{\star(K-2)}) \lneqq_{\mathcal{I}_{K-1}}  {\mathbf{z}}^{\star(K-1)}$ $(\exists\mathcal{I}_{K-1}\subset\{1,...,N_{K-1}\})$. 
Then, for $\forall i_{K-1}\in\mathcal{I}_{K-1}$, we can choose $\varepsilon_{i_{K-1}}  \in (0,1)$ that satisfies  
\begin{align}
	[\mathcal{R}^{(K-2)}({\mathbf{z}}^{\star(K-2)})]_{i_{K-1}} < \varepsilon_{i_{K-1}} {z}_{i_{K-1}}^{\star(K-1)} < {z}_{i_{K-1}}^{\star(K-1)}. 
\end{align}
Define $\widecheck{\mathbf{z}}^{(K-1)}$ as $\widecheck{{z}}^{(K-1)}_{i_{K-1}} = \varepsilon_{i_{K-1}} {z}_{i_{K-1}}^{\star(K-1)}$ and $\widecheck{{z}}^{(K-1)}_{i} = {z}_{i}^{\star(K-1)}$ ($i\neq i_{K-1}$), then $\widecheck{\mathbf{z}}^{(K-1)}\lneqq_{\{i_{K-1}\}}{\mathbf{z}}^{(K-1)}$.
Since $\mathcal{R}^{(K-1)}$ has the strictly increasing property, there exists one index $i_{K} \in \{1,\ldots,N_K\}$ such that $[\mathcal{R}^{(K-1)}(\widecheck{\mathbf{z}}^{(K-1)})]_{i_K}<[\mathcal{R}^{(K-1)}({{\mathbf{z}}}^{\star(K-1)})]_{i_K}={{z}}^{\star(K)}_{i_K}$, and we can choose $\varepsilon_{i_K}  \in (0,1)$ satisfying 
\begin{align} [\mathcal{R}^{(K-1)}(\widecheck{\mathbf{z}}^{(K-1)})]_{i_K} <  \varepsilon_{i_K} {z}_{i_K}^{\star(K)} < {z}_{i_K}^{\star(K)}.
\end{align}
Define $\widecheck{\mathbf{z}}^{(K)}$ as in \eqref{eq:ztild1}, then, by the same discussion, we can derive that
$
\mathcal{R}^{(K)}(\widecheck{\mathbf{z}}^{(K)}) < \mathcal{R}^{(K)}({{\mathbf{z}}}^{\star(K)})
$ and its contradiction. By repeating the same discussion, we can verify Lemma \ref{Lemma1}.

\section{Proof of Proposition \ref{Theorem1}}\label{ap:poof_ER}
First, we show $\mathcal{S}_\x\subset\widetilde{\mathcal{S}_\x}$ by contradiction. Suppose that there exists $\x^\star\in \mathcal{S}_\x \Rightarrow \x^\star\notin \widetilde{\mathcal{S}_\x}$. Let us define $\{\mathbf{z}^{\star(k)}\}_{k=K}^{2}$ as follows:
\begin{align}
	\mathbf{z}^{\star(k+1)}&:=\mathcal{R}^{(k)}(\mathbf{z}^{\star(k)})\ (2\leq k\leq K-1),\nonumber\\
	\mathbf{z}^{\star(\textcolor{black}{2})}&:=\mathcal{R}^{(1)}(\mathbf{A}\mathbf{x}^\star),
\end{align}
then, $( \mathbf{x}^\star, \{\mathbf{z}^{\star(k)}\}_{k=K}^{2})\in\mathcal{D}$ is the minimizer of \eqref{eq:reform2}. From the assumption ($( \mathbf{x}^\star, \{\mathbf{z}^{\star(k)}\}_{k=K}^{2})\notin\widetilde{\mathcal{S}}_{\mathbf{x}}\times\prod_{k=K}^{2}\widetilde{\mathcal{S}}_{\mathbf{z}^{(k)}}$), there exists $( \widetilde{\mathbf{x}}, \{\widetilde{\mathbf{z}}^{(k)}\}_{k=K}^{2})\in\widetilde{\mathcal{S}}_{\mathbf{x}}\times\prod_{k=K}^{2}\widetilde{\mathcal{S}}_{\mathbf{z}^{(k)}}(\subset\widetilde{\mathcal{D}})$ such that
\begin{align}\label{eq:proof_P1_1}
	\mathcal{U}(\x^\star)=\mathcal{V}(\mathbf{x}^\star, \{\mathbf{z}^{\star(k)}\}_{k=K}^{2})>\widetilde{\mathcal{V}}(\widetilde{\mathbf{x}}, \{\widetilde{\mathbf{z}}^{(k)}\}_{k=K}^{2})=\mathcal{U}(\widetilde{\x}).
\end{align}
The last equality in \eqref{eq:proof_P1_1} follows $( \widetilde{\mathbf{x}}, \{\widetilde{\mathbf{z}}^{(k)}\}_{k=K}^{2})\in{\mathcal{D}}$ from Lemma \ref{Lemma1}. This contradicts that $\x^\star\in\mathcal{S}_\x$, and thus, $\mathcal{S}_\x\subset\widetilde{\mathcal{S}}_\x$.

Next, we show $\widetilde{\mathcal{S}_\x}\subset\mathcal{S}_\x$ by contradiction similarly. Suppose that there exists $\x^\star\in\widetilde{\mathcal{S}_\x} \Rightarrow \x^\star\notin \mathcal{S}_\x$. From $\mathcal{D}\subset\widetilde{\mathcal{D}}$, $( \x^\star, \{\mathbf{z}^{\star(k)}\}_{k=K}^{2})\in\widetilde{\mathcal{S}}_{\mathbf{x}}\times\prod_{k=K}^{2}\widetilde{\mathcal{S}}_{\mathbf{z}^{(k)}}$ satisfies:
\begin{align}\label{eq:proof_P1_2}
	&\min_{( \mathbf{x}, \{\mathbf{z}^{(k)}\}_{k=K}^{2}) \in \mathcal{D} } \mathcal{V}(\mathbf{x}, \{\mathbf{z}^{(k)}\}_{k=K}^{2})\geq\nonumber\\
	&\min_{( \mathbf{x}, \{\mathbf{z}^{(k)}\}_{k=K}^{2}) \in \widetilde{\mathcal{D}} } \widetilde{\mathcal{V}}(\mathbf{x}, \{\mathbf{z}^{(k)}\}_{k=K}^{2})=\widetilde{\mathcal{V}}(\x^\star, \{\mathbf{z}^{\star(k)}\}_{k=K}^{2}).
\end{align}
Since $(\x^\star, \{\mathbf{z}^{\star(k)}\}_{k=K}^{2})\in\mathcal{D}$ follows from Lemma \ref{Lemma1}, $\forall( {\mathbf{x}}, \{{\mathbf{z}}^{(k)}\}_{k=K}^{2})\in\mathcal{D}$ satisfies
\begin{align}\label{eq:proof_P1_3}
	&\mathcal{U}(\x)=\mathcal{V}(\mathbf{x}, \{\mathbf{z}^{(k)}\}_{k=K}^{2})\geq{\mathcal{V}}(\x^\star, \{\mathbf{z}^{\star(k)}\}_{k=K}^{2})=\mathcal{U}(\x^\star).
\end{align}
This contradict the assumption $\x^\star\notin\mathcal{S}_\x$. Therefore, $\widetilde{\mathcal{S}}_\x\subset\mathcal{S}_\x$, and we complete the proof for Proposition \ref{Theorem1}\footnote{The detail of the proof for Theorem 1 in \cite{Kyochi2021} can also be verified by the same discussion in Appendix A.}.

\section{Proof of Proposition \ref{Theorem2}}\label{ap:poof_ERLiGME}
	\begin{enumerate}
		\item If $\mathbf{z}_{\mathrm{R}} = \mathbf{0}$, then Problem \eqref{eq:prob1cELG} is formulated as follows:
		\begin{align}\label{eq:T2_proof1}
			&\argmin_{\mathbf{x},\{\mathbf{z}^{(k)}\}_{k=K}^{2}} \frac{1}{2} \| \boldsymbol{\Phi}\mathbf{x} - \mathbf{y} \|_2^2 + \frac{\rho}{2} \| \mathbf{z}^{(K)} \|_2^2 + \mu \mathcal{R}_{\mathbf{B}}^{(K)}(\mathbf{z}^{(K)}), \nonumber\\
			&\ \mathrm{s.t.}\ \ \mathcal{R}^{(1)}(\mathbf{A}\mathbf{x}) \leq \mathbf{z}^{(2)},\ \mathcal{R}^{(k)}(\mathbf{z}^{(k)}) \leq \mathbf{z}^{(k+1)},
		\end{align}
		where $2\leq k\leq K-1$. Define $\widetilde{\mathcal{R}}^{(K)}(\mathbf{z}^{(K)}) := \frac{\rho}{2}\| \mathbf{z}^{(K)} \|_2^2 + \mu \mathcal{R}_\B^{(K)}(\mathbf{z}^{(K)})$, then $\widetilde{\mathcal{R}}^{(K)}(\mathbf{z}^{(K)}) $ is a strictly-increasing function. From the same discussion in the proof for Proposition 1, we can verify that $\mathcal{S}_{\mathbf{x}} = \widetilde{\mathcal{S}}_{\mathbf{x}}$.
		
		\item 
		First, we prove $( {\mathbf{x}}^\star, \{{\mathbf{z}}^{\star(k)}\}_{k=K}^{2})\in{\mathcal{D}}$ by contradiction. Suppose $\mathcal{R}^{(k)}(\mathbf{z}^{\star(k)}) \neq \mathbf{z}^{\star(k+1)}$ for $k=K-1$. From the assumption ($\mathbf{z}_{\mathrm{R}} \neq \mathbf{0}$ and $[\mathbf{z}_{\mathrm{R}}]_i < [{\mathbf{z}}^{\star(K)}]_i$ holds for any $i$ such that $[\mathbf{z}_{\mathrm{R}}]_i \neq 0$), there exists an index set $\mathcal{I}_K\subset\{1,...,N_K\}$ that satisfies $\mathcal{R}^{(K-1)}({\mathbf{z}}^{\star(K-1)})\lneqq_{\mathcal{I}_K}{{\z}}^{\star(K)}$. For $\forall i_K\in\mathcal{I}_K$, we can choose $\varepsilon_{i_K}  \in (0,1)$ satisfying
		\begin{align}
			0\leq\max\{[\mathcal{R}^{(K-1)}(&\mathbf{z}^{\star(K-1)})]_{i_K}, [\mathbf{z}_{\mathrm{R}}]_{i_K}\} \nonumber\\&\quad< \varepsilon_{i_K} {z}_{i_K}^{\star(K)} < {z}_{i_K}^{\star(K)}.
		\end{align}
		Define $\widecheck{\mathbf{z}}$ as $\widecheck{{z}}^{(K)}_{i_{K}} = \varepsilon_{i_{K}} {z}_{i_{K}}^{\star(K)}$ and $\widecheck{{z}}^{(K)}_{i} = {z}_{i}^{\star(K)}$ ($i\neq i_{K}$),  then $\widecheck{\mathbf{z}}^{(K)}\lneqq_{\{i_{K}\}}{\mathbf{z}}^{\star(K)}$.	Since $\mathcal{R}^{(K)}$ has the strictly increasing property, $\mathcal{R}^{(K)}(\widecheck{\mathbf{z}}) < \mathcal{R}^{(K)}(\mathbf{z}^{\star(K)})$ holds. Moreover,
		\begin{align}
			&\| \mathbf{z}^{\star(K)} - \mathbf{z}_{\mathrm{R}}\|_2^2 \nonumber\\
			>& \sum_{i \neq i_K}^{N_K} ( [\mathbf{z}^{\star(K)}]_i - [\mathbf{z}_{\mathrm{R}}]_i )^2 + ( \varepsilon_{i_K} {z}_{i_K}^{\star(K)} - [\mathbf{z}_{\mathrm{R}}]_{i_K} )^2\nonumber\\
			&=\ \| \widecheck{\mathbf{z}} - \mathbf{z}_{\mathrm{R}}\|_2^2.
		\end{align}    
		These contradict that $\mathbf{z}^{\star(K)}$ is the minimizer of \eqref{eq:T2_LiGME2}. For $2\leq k \leq K-2$, the same reasoning in the proof of Lemma \ref{Lemma1} can be derived. Therefore, $( {\mathbf{x}}^\star, \{{\mathbf{z}}^{\star(k)}\}_{k=K}^{2})$ satisfies the following equations:
		\begin{align}\label{eq:lemma1}
			\mathcal{R}^{(1)}(\mathbf{A}{\mathbf{x}}^\star) = {\mathbf{z}}^{\star(\textcolor{black}{2})},\quad\mathcal{R}^{(k)}({\mathbf{z}}^{\star(k)})=\textcolor{black}{{\mathbf{z}}^{\star(k+1)},}
		\end{align}
		where $2\leq k\leq K-1$, and thus, $( {\mathbf{x}}^\star, \{{\mathbf{z}}^{\star(k)}\}_{k=K}^{2})\in{\mathcal{D}}$.
		Then, by the same discussion in the proof of Proposition \ref{Theorem1}, we can verify $\mathcal{S}_{\mathbf{x}} = \widetilde{\mathcal{S}}_{\mathbf{x}}$.
	\end{enumerate}

\section{Examples of Non-Proximable Multilayered Regularizers}\label{ap:examples_MR}
\subsection{Decorrelated Structure-tensor Total Variation}\label{subsec:EpiDSTV}
For a vectorized $N\times N$ color image $\mathbf{x} \in \mathbb{R}^{3N^2}$, 
the DSTV $\|\mathbf{x}\|_{\mathrm{DSTV}}: \mathbb{R}^{3N^2} \rightarrow \mathbb{R}_{++}$ with $W \times W$ patches
is defined as:
\begin{align}
	&\|\mathbf{x}\|_{\mathrm{DSTV}} = \sum^{N^2}_{n=1} \frac{1}{2}\|\boldsymbol{\mathfrak{D}}_{n}^{(\mathrm{y})}\|_\ast +\sum^{N^2}_{n=1} \sqrt{\sum^{2}_{c=1}\|\boldsymbol{\mathfrak{D}}_{n}^{(c)}\|_\ast^2},\nonumber\\
	&\boldsymbol{\mathfrak{D}}_{n}^{(\mathrm{y})} =\begin{bmatrix}
		\mathbf{d}_{n}^{(\mathrm{y, v})} & \mathbf{d}_{n}^{(\mathrm{y,h})}
	\end{bmatrix},\ \boldsymbol{\mathfrak{D}}_{n}^{(c)} = \begin{bmatrix}
		\mathbf{d}_{n}^{(c,\mathrm{v})} & \mathbf{d}_{n}^{(c,\mathrm{h})}
	\end{bmatrix},
\end{align}
where $\mathbf{d}_{n}^{(\mathrm{y, v})} \in \mathbb{R}^{W^2}$ and $\mathbf{d}_{n}^{(\mathrm{y, h})} \in \mathbb{R}^{W^2}$ are the vertical and horizontal difference in the $W \times W$ local region centered at the $n$-th luma sample, and $\mathbf{d}_{n}^{(c, \mathrm{v})}$ and $\mathbf{d}_{n}^{(c, \mathrm{h})}$ are those of the chroma sample. 
The DSTV can be formulated by using a (possibly non-proximable) 3-layered mixed norm as follows: 
\begin{align}\label{eq:defDSTV}
	&\|\mathbf{x}\|_{\mathrm{DSTV}}= \|\mathbf{PED}\mathbf{C}_{\mathrm{D}}\mathbf{x}\|_{\ast, 2, 1}^{(W^2,1,2)},\nonumber\\
	&\|\cdot \|_{\ast, 2, 1}^{(W^2,1,2)}:= \| \cdot \|_{2, 1}^{(1,2)} \circ \|\cdot\|_\ast^{(W^2)}: \mathbb{R}^{6W^2N^2} \rightarrow \mathbb{R}_{++},\nonumber\\
	&\| \cdot \|_{\ast}^{(W^2)}: \mathbb{R}^{6W^2N} \rightarrow \mathbb{R}_{++}^{3N^2},\ \mathbf{x} = \begin{bmatrix}
		\mathbf{x}_1^{\top} ,\ \cdots,\ \mathbf{x}_{3N^2}^{\top}
	\end{bmatrix}^\top \mapsto\nonumber\\
	&\|\mathbf{x}\|_\ast^{(W^2)} = \begin{bmatrix}
		\| \mathrm{vec}_{(W^2,2)}^{-1}(\mathbf{x}_1) \|_\ast,\  \cdots,\ \| \mathrm{vec}_{(W^2,2)}^{-1}(\mathbf{x}_{3N^2}) \|_\ast
	\end{bmatrix}^\top,\nonumber\\
	&\| \cdot \|_{2, 1}^{(1,2)}: \mathbb{R}^{3N} \rightarrow \mathbb{R}_{++},\ \mathbf{x} \mapsto \|\mathbf{x}\|_{2, 1}^{(1,2)} = \frac{1}{2} \| \mathbf{x}_{\mathrm{y}} \|_{2,1}^{(1)} + \|\mathbf{x}_c \|_{2,1}^{(2)}\nonumber\\
	& (\mathbf{x}_{\mathrm{y}} \in \mathbb{R}^{N},\ \mathbf{x}_c\in \mathbb{R}^{2N}),
\end{align}
where $\mathbf{C}_{\mathrm{D}}=\mathbf{C}_0 \otimes \mathbf{I}_{[N^2]} \in \mathbb{R}^{3N^2 \times 3N^2}$, $\mathbf{C}_0 \in \mathbb{R}^{3\times 3}$ is the DCT matrix to convert the RGB color space to the luma-chroma one, $\mathbf{D} = \mathrm{diag}(\mathbf{D}_{\mathrm{vh}}, \mathbf{D}_{\mathrm{vh}}, \mathbf{D}_{\mathrm{vh}}) \in \mathbb{R}^{6N^2 \times 3N^2}$, $\mathbf{D}_{\mathrm{vh}} = \begin{bmatrix}
	\mathbf{I}_{[N^2]} \otimes \mathbf{D}_0^\top & \mathbf{D}_0^\top\otimes \mathbf{I}_{[N^2]} 
\end{bmatrix}^\top  \in \mathbb{R}^{2N^2 \times N^2}$ is the vertical and horizontal difference matrix, $\mathbf{D}_0$ is the difference operator extended with zeros in the last row, i.e.,
\begin{align}
	\mathbf{D}_0 := \begin{bmatrix}
		1 & -1 & & \\
		& \ddots & \ddots & \\
		&  & 1 & -1 \\
		0 & 0 & ... & 0 \\
	\end{bmatrix}\in \R^{N\times N},
\end{align}
and $\mathbf{E} \in \mathbb{R}^{6W^2N^2 \times 6N^2}$ is an expansion operator that duplicates all the gradients $\mathbf{PD}\mathbf{C}_{\mathrm{D}}\mathbf{x}$ in fully overlapping patches with applying weights to each location in the patch, $\mathbf{P}$ denotes a permutation matrix. 

\subsection{Amplitude Spectrum Nuclear Norm}
\label{subsec:ASNN}
For a matrix $\mathbf{L} \in \mathbb{R}^{M\times N}$, 
the ASNN regularizer $\|\mathbf{L}\|_{\mathrm{ASNN}}: \mathbb{R}^{M\times N} \rightarrow \mathbb{R}_{++}$ is defined as 
\begin{align}\label{eq:defASNN}
	&\|\mathbf{L}\|_{\mathrm{ASNN}} = \| | \mathbf{WL} | \|_* = \| | 
	(\mathbf{W}^{[c]} - j\mathbf{W}^{[s]}) \L
	| \|_*,
\end{align}
where $|\cdot |$ is the element-wise absolute value operation,  $\mathbf{W}\in\mathbb{C}^{M\times M}$ is the normalized DFT matrix ($[\mathbf{W}]_{m,n} := \frac{1}{\sqrt{M}}\exp{(j\frac{2\pi}{M}mn)}$, $j = \sqrt{-1}$), and $\mathbf{W}^{[c]}, \mathbf{W}^{[s]}$ are the real and imaginary parts of $\mathbf{W}$ respectively. 
The ASNN can be formulated by a (possibly non-proximable) 2-layered mixed norm for $\boldsymbol{\ell}=\mathrm{vec}(\mathbf{L})$ as follows: 
\begin{align}\label{eq:defASNN2}
	&\|\boldsymbol{\ell}\|_{\mathrm{ASNN}} = \| \mathbf{T}\boldsymbol{\ell} \|_{2,*}^{(M,N)},\nonumber\\
	& \| \cdot \|_{2,*}^{(M,N)} := \| \mathrm{vec}^{-1}_{(M,N)}(\cdot) \|_{*} \circ  \| \cdot \|_{2}^{(M,N)} : \mathbb{R}^{2MN} \rightarrow \mathbb{R}_{++},\nonumber\\
	& \| \cdot \|_{2}^{(M,N)}: \mathbb{R}^{2MN} \rightarrow \mathbb{R}^{MN},\nonumber\\
	&\mathbf{x}:= \begin{bmatrix}
		\mathbf{x}_{1,1}^\top & ... & \mathbf{x}_{M,1}^\top & ... & ... & \mathbf{x}_{1,N}^\top & ... & \mathbf{x}_{M,N}^\top
	\end{bmatrix}^\top \mapsto\nonumber\\
	&\| \mathbf{x} \|_{2}^{(M,N)} = \begin{bmatrix}
		\|\mathbf{x}_{1,1}\|_2 & ... & \|\mathbf{x}_{1,N}\|_2 
		& ... & ... & \|\mathbf{x}_{M,N}\|_2
	\end{bmatrix}^\top,\nonumber\\
	&\mathbf{T}\boldsymbol{\ell} = \mathbf{Q}\begin{bmatrix}
	\I_{[N]}\otimes\mathbf{W}^{[c]} \\ -\I_{[N]}\otimes \mathbf{W}^{[s]}
	\end{bmatrix}\boldsymbol{\ell}\nonumber\\
	&\hspace{1.5em}= \begin{bmatrix}
	\mathbf{t}_{1,1}^\top & ... & \mathbf{t}_{M,1}^\top & ... & ... & \mathbf{t}_{1,N}^\top & ... & \mathbf{t}_{M,N}^\top
	\end{bmatrix}^\top,\hspace{-0.5em}
\end{align}
where $\mathbf{t}_{m,n},\,\mathbf{x}_{m,n}\in\R^2$ and $\mathbf{Q}$ is a permutation matrix.

\section{Applying for Single-Layered Regularization}\label{ap:singleReg}
The ER-LiGME model in \eqref{eq:prob1cELG} for a single-layered regularizer $\mathcal{R}^{(1)}(\mathbf{A}\mathbf{x})$ is defined by the following settings:
\begin{align}
&\argmin_{\widehat{\mathbf{x}}}\frac{1}{2}\| {\widehat{\boldsymbol{\Phi}}}\widehat{\mathbf{x}} - \widehat{\mathbf{y}} \|_2^2 + \mu 
\widehat{\mathcal{R}}_{\mathbf{B}}^{(1)}(\widehat{\mathbf{x}}) + \iota_{\mathrm{Id} \cap \mathcal{C} }(\mathbf{C}\widehat{\mathbf{x}}), \nonumber\\
&\widehat{\boldsymbol{\Phi}} = \mathrm{diag}(\boldsymbol{\Phi},  \sqrt{\rho}\mathbf{I} ),\ \widehat{\mathbf{x}} = \begin{bmatrix} \mathbf{x}^\top & \mathbf{z}^{\top}\end{bmatrix}^\top,\nonumber \\
&\mathbf{C}\widehat{\mathbf{x}}  =\begin{bmatrix} (\mathbf{A}\mathbf{x})^\top & \mathbf{z}^{\top} & \mathbf{x}^{\top} \end{bmatrix}^\top, 
\end{align}
where $\mathrm{Id} = \{ (\mathbf{x}, \mathbf{y})\ |\ \mathbf{x}= \mathbf{y} \}$. Since $\mathrm{Id}$ is convex, the relaxation is not required and, in optimization, the projection to $\mathrm{Id}$ ($(\mathbf{x}^\star, \mathbf{y}^\star) = \mathcal{P}_{\mathrm{Id}}(\mathbf{x}, \mathbf{y})$) can be easily computed as
\begin{align}
([\mathbf{x}^\star]_n, [\mathbf{y}^\star]_n) =&  \begin{cases}
({x}_n, {y}_n) & (({x}_n, {y}_n) \in \mathrm{Id}) \\
\left( \frac{x_n+y_n}{2}, \frac{x_n+y_n}{2} \right) & (\mathrm{otherwise})
\end{cases}.
\end{align}
By setting $\mathbf{B} = \mathrm{diag}\left(\mathbf{O}, \sqrt{\frac{\theta}{\mu}}\mathbf{I}\right)$ with the parameters $0\leq \theta \leq \rho \leq \infty$ as in \eqref{eq:settingB}, the overall convexity condition can be satisfied easily. 



\ifCLASSOPTIONcaptionsoff
\newpage
\fi



%

	
	

\bibliographystyle{IEEEtran}
\end{document}